\def\compilelong{}

\ifx\compilelong\undefined
\documentclass[sigconf,screen]{acmart}
\else
\documentclass[11pt,twoside]{article}
\fi

\usepackage[utf8]{inputenc}
\usepackage{jeffe}

\ifx\compilelong\undefined
\else
\usepackage[margin=1in]{geometry}
\fi

\usepackage{graphicx}
\usepackage{enumitem}

\ifx\compilelong\undefined
\usepackage{balance}
\else
\usepackage{cite}
\fi

\usepackage{verbatim}
\usepackage{gensymb}

\usepackage{microtype}
\ifx\compilelong\undefined
\else
\usepackage[charter]{mathdesign}
\usepackage{berasans,beramono}
\fi
\usepackage[mathcal]{euscript}
\usepackage{stmaryrd}

\ifx\compilelong\undefined
\else
\usepackage[dvipsnames,cmyk]{xcolor}
\fi

\usepackage{pdfpages}

\usepackage{comment}
\ifx\compilelong\undefined
  \excludecomment{longversion}
  \newenvironment{shortversion}{}{}
\else
  \excludecomment{shortversion}
  \newenvironment{longversion}{}{}
\fi

\def\eps{\varepsilon}
\def\numfaces{|F|}

\def\concat{\circ}
\def\tuplecat{{+\!\!+}\,}

\def\arcto{\mathord\shortrightarrow}

\def\wallwith{\mathord\mid}

\def\fencewith{\mathord\shortuparrow}
\def\fence#1#2{#1\mathord\shortuparrow#2}

\def\dartof#1{\vec{#1}}
\def\dartsof#1{\dartof{#1}}

\def\rev{\mathit{rev}}

\def\permutation{\pi}

\def\cycle{\gamma}
\def\cycles{\Gamma}
\def\walk{P}
\def\corewalk{\bar{P}}
\def\tree{T}
\def\cotree{C}
\def\leftovers{L}
\def\fundcycle{\mathit{cycle}}

\def\cutgraph{X}
\def\core{\bar{\cutgraph}}
\def\hair{H}
\def\cutpath{\pi}
\def\cutpaths{\Pi}
\def\reducedcutgraph{\tilde{\cutgraph}}

\def\cutedge{\pi}
\def\cutdart{\xi}
\def\cutdarts{\vec{\Pi}}
\def\darttree{\mathcal{A}}

\def\facefinger{f}
\def\finger{\facefinger}

\def\flow{f}
\def\imbalance#1{\delta{#1}}

\def\coflow{z}
\def\coimbalance#1{\partial{#1}}
\def\cwneighbors{\partial^-}

\def\cost{c}
\def\capacity{\mu}
\def\demand{b}

\def\homsig#1{[#1]}
\def\potential{\alpha}
\def\dist{\mathit{dist}}
\def\parameter{\lambda}
\def\uparameter{\lambda_0}
\def\slack{slack}
\def\uslack{slack_0}
\def\shortestpath{\sigma}

\def\correspondence{\mathcal{C}}

\begin{longversion}
\newtheorem{theorem}{Theorem}[section]
\newtheorem{corollary}[theorem]{Corollary}
\newtheorem{lemma}[theorem]{Lemma}
\end{longversion}

\begin{shortversion}

\title{Holiest Minimum-Cost Paths and Flows in Surface Graphs}
\titlenote{ This work was initiated at Dagstuhl seminar 16221 ``Algorithms for Optimization Problems
in Planar Graphs''.
The latest full version of this paper can be found at
\url{https://utdallas.edu/\~kyle.fox/publications/holiest.pdf}.
The research presented in this paper was partially supported by
\grantsponsor{NSF}{NSF}{https://www.nsf.gov} grants \grantnum{NSF}{CCF-1408763},
\grantnum{NSF}{IIS-1408846}, \grantnum{NSF}{IIS-1447554}, \grantnum{NSF}{CCF-1513816},
\grantnum{NSF}{CCF-1546392}, \grantnum{NSF}{CCF-1527084}, and \grantnum{NSF}{CCF-1535972}; by
\grantsponsor{ARO}{ARO}{http://www.arl.army.mil/www/default.cfm?page=29} grant
\grantnum{ARO}{W911NF-15-1-0408}, and by grant \grantnum{USIBSF}{2012/229} from the
\grantsponsor{USIBSF}{U.S.-Israel Binational Science Foundation}{http://www.bsf.org.il/}.
}

\author{Jeff Erickson}
\affiliation{%
  \department{Computer Science}
  \institution{University of Illinois at Urbana-Champaign}
  \city{Urbana}
  \state{IL}
  \postcode{61801-2302}
  \country{USA}}
\email{jeffe@illinois.edu}

\author{Kyle Fox}
\affiliation{%
  \department{Computer Science}
  \institution{The University of Texas at Dallas}
  \city{Richardson}
  \state{TX}
  \postcode{75080}
  \country{USA}}
\email{kyle.fox@utdallas.edu}
\authornote{Portions of this work were done while the author was a postdoctoral associate at Duke
University.}

\author{Luvsandondov Lkhamsuren}
\affiliation{%
  \institution{Airbnb}
  \city{San Francisco}
  \state{CA}
  \postcode{94103}
  \country{USA}}
\email{lkhamsurenl@gmail.com}
\authornote{Portions of this work were done while this author was a student at the University of
Illinois at Urbana-Champaign.}
\end{shortversion}
\begin{document}

\ifx\compilelong\undefined
\else
\pagestyle{myheadings}
\markboth{Holiest Minimum-Cost Paths and Flows in Surface Graphs}
		 {Jeff Erickson, Kyle Fox, and Luvsandondov Lkhamsuren}
\begin{titlepage}
\title{Holiest Minimum-Cost Paths and Flows in Surface Graphs\footnote{
This work was initiated at Dagstuhl seminar 16221 ``Algorithms for Optimization Problems in Planar
Graphs''.
The latest full version of this paper can be found at
\url{https://utdallas.edu/\~kyle.fox/publications/holiest.pdf}.
The research presented in this paper was partially supported by NSF grants CCF-1408763, IIS-1408846,
IIS-1447554, CCF-1513816, CCF-1546392, CCF-1527084, and CCF-1535972; by ARO grant W911NF-15-1-0408,
and by grant 2012/229 from the U.S.-Israel Binational Science Foundation.
}
}

\author{
  Jeff Erickson%
  \thanks{Department of Computer Science,
    University of Illinois, Urbana-Champaign; \url{jeffe@illinois.edu}.
    }
  \and
  Kyle Fox%
  \thanks{Department of Computer Science,
    The University of Texas at Dallas; \url{kyle.fox@utdallas.edu}.
    Portions of this work were done while the author was a postdoctoral associate at Duke
    University.
    }
  \and
  Luvsandondov Lkhamsuren%
  \thanks{Airbnb; \url{lkhamsurenl@gmail.com}.
    Portions of this work were done while this author was a student at the University of Illinois at
    Urbana-Champaign.
    }
  }

\maketitle
\fi
\begin{abstract}
Let \(G\) be an edge-weighted directed graph with \(n\) vertices embedded on an orientable surface
of genus \(g\).
We describe a simple deterministic lexicographic perturbation scheme that guarantees uniqueness of
minimum-cost flows and shortest paths in~\(G\).
The perturbations take~\(O(gn)\) time to compute.
We use our perturbation scheme in a black box manner to derive a deterministic \(O(n \log \log n)\)
time algorithm for minimum cut in \emph{directed} edge-weighted planar graphs and a deterministic
\(O(g^2 n \log n)\) time proprocessing scheme for the multiple-source shortest paths problem of
computing a shortest path oracle for all vertices lying on a common face of a surface embedded
graph.
The latter result yields faster deterministic near-linear time algorithms for a variety of
problems in constant genus surface embedded graphs.

Finally, we open the black box in order to generalize a recent linear-time algorithm for
multiple-source shortest paths in unweighted undirected planar graphs to work in arbitrary
orientable surfaces.
Our algorithm runs in \(O(g^2 n \log g)\) time in this setting, and it can be used to give improved
linear time algorithms for several problems in unweighted undirected surface embedded graphs of
constant genus including the computation of minimum cuts, shortest topologically non-trivial cycles,
and minimum homology bases.
\end{abstract}
\ifx\compilelong\undefined
%%% The following is specific to STOC'18 and the paper
%%% 'Holiest Minimum-Cost Paths and Flows in Surface Graphs'
%%% by Jeff Erickson, Kyle Fox, and Luvsandondov Lkhamsuren.
%%%
\setcopyright{acmcopyright}
\acmPrice{15.00}
\acmDOI{10.1145/3188745.3188904}
\acmYear{2018}
\copyrightyear{2018}
\acmISBN{978-1-4503-5559-9/18/06}
\acmConference[STOC'18]{50th Annual ACM SIGACT Symposium on the Theory of Computing}{June 25--29, 2018}{Los Angeles, CA, USA}


\maketitle
\else
\noindent

\thispagestyle{empty}
\setcounter{page}{0}
\end{titlepage}
\fi

%---------------------------------------
\section{Introduction}
\label{sec:intro}

Many recent combinatorial optimization algorithms for directed surface embedded graphs rely on a
common assumption:
the shortest path between any pair of vertices is unique.
The most commonly applied consequence of this assumption is that the shortest paths entering (or
leaving) a common vertex do not cross one another.
From this consequence, one can prove near-linear running time bounds for a variety of problems,
including the computation of maximum flows~\cite{bk-amfdp-09,e-mfpsp-10,bkmnw-msmsm-17,ek-lafms-13}
and global minimum cuts~\cite{mnnw-mcdpg-18} in directed planar (genus \(0\)) graphs
as well as the computation of minimum cut oracles in planar and more general embedded
graphs~\cite{bsw-mscop-15,benw-apmcn-16} (see also Wulff-Nilsen~\cite{w-mcbap-09}).

This assumption is also used in algorithms for the multiple-source shortest paths problem introduced
for planar graphs by Klein~\cite{k-msspp-05}.
In the multiple-source shortest paths problem, one is given a surface embedded graph \(G = (V, E,
F)\) of genus \(g\) with vertices \(V\), edges \(E\), and faces \(F\).
The goal is to compute a representation of all shortest paths from vertices on a common face \(r \in
F\) to all other vertices in the graph.
Assuming uniqueness of shortest paths, multiple-source shortest paths can be computed in only~\(O(g
n \log n)\) time~\cite{k-msspp-05,cce-msspe-13}.
Algorithms for this problem can be used to solve a variety of problems in planar and more general
surface embedded graphs of constant genus in near-linear time.
Such results include the computation of shortest cycles with non-trivial
topology~\cite{cce-msspe-13,ew-csec-10,en-mcsnc-11,e-sncds-11,f-sntcd-13, bcfn-mchbs-17}, the
computation of maximum flows and minimum
cuts~\cite{en-mcsnc-11,insw-iamcmf-11,bkmnw-msmsm-17,efn-gmcse-12,cen-hfcc-12,lnsw-ssasm-12,cl-cgpgl-13}, the
computation of exact and approximate
distance oracles~\cite{c-mdpg-10,kks-lsado-11,ms-edopg-12}, and even the computation of \emph{single}-source shortest
paths~\cite{kmw-spdpg-09,mw-sppgr-10}.

\paragraph*{Enforcing uniqueness.}
Unfortunately, it is often difficult to actually enforce the assumption that shortest paths are
unique.
One popular method is to add tiny random perturbations to the lengths of edges, and then apply a
variant of the Isolation Lemma of Mulmuley \etal~\cite{mvv-memi-87} to argue that shortest paths are
unique \emph{with high probability}.
This method is used directly by Erickson~\cite{e-mfpsp-10},
Mozes~\etal~\cite{mnnw-mcdpg-18}, Cabello \etal~\cite{cce-msspe-13}, and the numerous
papers that rely on the latter result.
%In turn, it is also used in numerous papers that rely on the multiple-source shortest paths data
%structure of Cabello \etal~\cite{cce-msspe-13} for surfaces of non-zero genus.

As an alternative to using randomness, one can instead use a \emph{lexicographic perturbation
scheme} where one redefines edge lengths to be multidimensional vectors so that comparisons can be
done lexicographically.
One such scheme was proposed by Charnes~\cite{c-odlp-52} and Dantzig \etal~\cite{dow-gsmml-55}, and
variants of it have been used for computing minimum cut oracles in planar
graphs~\cite{hm-apmcb-94,w-mcbap-09,bsw-mscop-15}.
In short, the scheme turns every edge length into an \(n+1\)-dimensional vector where \(n\) is the
number of edges in the graph.
The first component of the vector is the true length of the edge, but then there is a single other
component set to \(1\) based purely on the edge for which we are reassigning the length.
Naively implementing the scheme adds an~\(O(n)\) time overhead to all operations involving edge
length.
There are faster ways to use the scheme depending on the application one has in mind.
In particular, Cabello \etal~\cite{cce-msspe-13} implement the scheme with only a \(\log n\) factor
increase in the running time of their multiple-source shortest paths algorithm.
However, these fast implementations require some fairly heavy machinery, and even implementing
Dijkstra's~\cite{d-ntpcg-59} algorithm for single-source shortest paths requires that same \(\log
n\) factor increase in the running time and the use of relatively complex dynamic tree data
structures~\cite{t-dtste-97,hk-rfdga-99,tw-satt-05};
see Cabello \etal~\cite[Section 6.2]{cce-msspe-13}.
%Unfortunately, the scheme uses \(n\)-dimensional vectors, and implementing it efficiently even to
%run Dijkstra's algorithm~\cite{d-ntpcg-59} requires relatively complex dynamic tree data
%structures~\cite{t-dtste-97,hk-rfdga-99,tw-satt-05} and a \(\log n\) loss in asymptotic running
%time.
%Cabello~\etal~\cite{cce-msspe-13} suffer the same time loss in the deterministic version of their
%algorithm.

\paragraph*{Parametric shortest paths and the leafmost rule.}
Several algorithms for multiple-source shortest paths in embedded
graphs~\cite{k-msspp-05,cce-msspe-13,ek-lafms-13} and maximum flows in planar
graphs~\cite{bk-amfdp-09,e-mfpsp-10,ek-lafms-13} rely (at least by some interpretations) on the
parametric shortest paths framework introduced by Karp and Orlin~\cite{ko-pspaa-81,yto-fpspm-91}.
In short, these algorithms redefine the length of a subset of edges to increase or decrease by an
amount equal to some parameter \(\parameter\).
The algorithms then continuously increase \(\parameter\) while maintaining a shortest path tree
\(T\).
At certain values of \(\parameter\), an edge will \emph{pivot} into \(T\) while another edge pivots
out.
Uniqueness of shortest paths guarantees the total number of pivots to be small for the algorithms
mentioned above.
%One can use the unique shortest path assumption to bound the total number of pivots performed by
%the algorithms mentioned above.

That said, one can sometimes avoid the need for unique shortest paths by utilizing properties of
planar embeddings.
%That said, sometimes one can take advantage of planar embeddings to avoid the need for unique
%shortest paths altogether.
Klein~\cite{k-msspp-05}, Borradaile and Klein~\cite{bk-amfdp-09}, and Eisenstat and
Klein~\cite{ek-lafms-13} all give efficient algorithms that successfully use the parametric shortest
path framework without doing anything explicit to the edge lengths to guarantee unique shortest paths.
In particular, Eisenstat and Klein~\cite{ek-lafms-13} give linear-time maximum flow and
multiple-source shortest paths algorithms that \emph{cannot} take advantage of the perturbation
schemes mentioned above, because they crucially rely on the edge capacities/lengths being small
non-negative integers.
(Weihe~\cite{w-edstp-97} also describes a linear-time maximum flow algorithm for unweighted
undirected planar graphs, and Brandes and Wagner~\cite{bw-ltaad-00} give an algorithm for unweighted
\emph{directed} planar graphs.)
%See also Weihe~\cite{w-edstp-97} and Brandes and Wagner~\cite{bw-ltaad-00}.

Instead of using perturbation schemes, these algorithms all take advantage of the \emph{leafmost
rule} for selecting edges to pivot into the shortest path tree \(T\).
The leafmost rule works as follows:
The edges outside of \(T\) form a spanning tree \(\cotree\) of the planar dual graph.
Consider rooting \(\cotree\) at some dual vertex (primal face);
the root we choose depends upon the particular algorithm we are attempting to implement.
When \(\parameter\) reaches a value that requires pivoting an edge into \(T\) but there are multiple
appropriate candidate edges to choose from, the leafmost rule dictates that we should always select
the candidate edge lying closest to a leaf of \(\cotree\).
As a result, these algorithms all maintain \emph{leftmost} shortest path trees, assuming the initial
shortest path tree was itself leftmost.
We note the leafmost rule bears a strong resemblance to Cunningham's~\cite{c-nsm-76} rule for
maintaining a \emph{strongly feasible basis} during network simplex.

Despite these successes, the leafmost rule and leftmost shortest path trees still do not present an
ideal solution for algorithms requiring unique shortest paths.
For one, these algorithms need to be designed with leftmost shortest path trees in mind.
In contrast, random perturbations and lexicographic perturbation schemes can be implemented with only
minor changes in how comparisons and basic arithmetic operations are performed.
And perhaps more seriously, there is no obvious generalization of leftmost shortest path trees or
the leafmost rule for pivots in surface embedded graphs of non-zero genus.
In particular, the complement of a spanning tree is not itself a tree in this case.
Certain algorithms such as the multiple-source shortest paths algorithm of
Cabello~\etal~\cite{cce-msspe-13} appear to crucially rely on a guarantee that shortest paths really
are unique.

\subsection{Our Results}
\label{subsec:intro_results}

Let \(G\) be a graph of size \(n\) embedded in an orientable surface of genus \(g\) with lengths on
the edges.
We present a deterministic lexicographic perturbation scheme that guarantees uniqueness of shortest
paths despite using only \(O(g + 1)\)-dimensional vectors for the perturbed edge lengths.
The perturbation terms we use are all integers of absolute value \(O(n)\), so our scheme can be
employed in any combinatorial algorithm implemented in the word RAM model.
Using our scheme increases the asymptotic running time of such algorithms by at most a factor of
\(g\).

As detailed in Section~\ref{sec:scheme}, the perturbation vectors can be computed in \(O(gn)\)
time using a simple algorithm.
In short, we compute a \(2g\)-bit signature \([e]\) for each edge \(e\) so that the sum of edge
signatures along a cycle characterizes the \emph{homology class} of that cycle with coefficients in
\(\Z\).
A cycle's homology class describes how it wraps around the holes on a surface.
We also compute a single integer \(\coflow(e)\) for each edge \(e\) so that given a cycle \(\cycle\)
bounding a subset of faces \(F' \subseteq F\), the absolute value of the sum of these integers along
\(\cycle\) is equal to the number of faces in \(F'\).
%\begin{longversion}
This latter assignment of integers is inspired in part by results of Park and
Philips~\cite{pp-fmcpg-93} and Patel~\cite{p-deeoc-13} on the minimum quotient and sparsest cut
problems in planar graphs.
%\end{longversion}
Our perturbation vectors contain both \([e]\) and \(\coflow(e)\), and it is not difficult to show
that every (simple) cycle in \(G\) has non-zero cost according to our lexicographic perturbation
scheme.
Uniqueness of shortest paths follows as an easy consequence.

In fact, our scheme can be used to modify the \emph{costs} of edges in the more general minimum-cost
flow problem, guaranteeing that the minimum-cost flow itself is unique.
It turns out that our scheme encourages the selection of leftmost shortest paths or minimum-cost
flows in planar graphs, so we refer to the unique optimal solutions to these two problems as
\emph{homologically lexicographic least leftmost} minimum-cost paths and flows, or \emph{holiest}
paths and flows, for short.
%Our scheme can be immediately applied in a black box fashion to derandomize the recent \(O(n \log
%\log n)\) time minimum cut algorithm for directed planar graphs by
%Mozes~\etal~\cite{mnnw-mcdpg-18}.

%\begin{longversion}
After describing our perturbation scheme for computing holiest paths and flows, we turn to its
applications.
Using our scheme in a black box manner, we immediately derandomize the recent \(O(n \log \log n)\)
time minimum cut algorithm for directed planar graphs by
Mozes~\etal~\cite{mnnw-mcdpg-18}.%
\footnote{We admit that Mozes~\etal~were aware of the current work as they were writing their paper,
so they may not have felt a strong need to derandomize their algorithm themselves.}
%\end{longversion}

Our scheme can also be used in the multiple-source shortest paths algorithm of
Cabello~\etal~\cite{cce-msspe-13} for arbitrary surface embedded graphs, bringing its total running
time to \(O(g^2 n \log n)\).
Compared to the deterministic perturbation scheme they consider, our alternative provides a factor
\((\log n) / g\) improvement in running time, and the implementation is considerably simpler.
In turn, we obtain the same \((\log n) / g\) factor improvement to the deterministic versions of
nearly every algorithm that uses their data structure.
Cabello~\etal~actually require a slightly stronger condition than mere uniqueness of shortest paths,
but we are able to show our scheme guarantees the condition holds in Section~\ref{sec:MSSP}.
The exposition in that section also helps set us up for our remaining results.

It turns out that holiest paths and flows are not only leftmost objects of minimum-cost, but our
perturbation scheme also forces the aforementioned parametric shortest path based algorithms to
choose leafmost edges during pivots.
\begin{longversion}
See Section~\ref{sec:leafmost}.
\end{longversion}
Based on this observation, we generalize the linear-time multiple-source shortest paths algorithm of
Eisentstat and Klein~\cite{ek-lafms-13} for small integer edge lengths so that it works in surface
embedded graphs of arbitrary genus.
Our generalization runs in \(O(g(g n \log g + L))\) time where \(L\) is the sum of the integer edge
lengths.
Like Eisenstat and Klein, we must assume every edge has a reversal, essentially modeling unweighted
undirected graphs in the case that edge lengths are all \(1\).

The high level idea behind our algorithm is to generalize the leafmost rule using our new
perturbation scheme.
When we must pivot an edge into the holiest shortest path tree \(T\), we partition the set of
candidate edges into collections based on the homology class of their fundamental cycles with \(T\).
We pick a collection based on the homology signature portion of our scheme's perturbation vectors,
and then essentially apply the leafmost rule to edges \emph{within that collection} to select the
one that enters \(T\).
Finding the leafmost edge requires individually checking edges to see which ones can be pivoted into
\(T\).
Fortunately, we can charge the time spent checking these edges to changes in the \emph{homotopy}
class of the holiest paths to these edges' endpoints.
We give our algorithm and analysis in Section~\ref{sec:linear-time}.

Finally, using our linear-time algorithm for multiple-source shortest paths, we immediately obtain
new linear time algorithms for a variety of problems in unweighted undirected surface embedded
graphs, including the computation of \(s,t\)- and global minimum cuts, shortest cycles with
non-trivial embeddings, and shortest homology bases.
By combining known works, one can obtain linear time algorithms for each of these problems, assuming
the genus is a constant.
However, our new algorithms improve the running time for computing cuts from \(g^{O(g)} n\) to
\(2^{O(g)} n\), and they improve the running time for the other problems from \(2^{O(g)} n\) to
\(O(\poly(g) n)\).
In particular, ours are the first algorithms for the latter problems that simultaneously have
polynomial dependency on \(g\) and linear dependency on \(n\).
%\begin{longversion}
We describe these applications in Section~\ref{sec:applications}.
%\end{longversion}

\begin{shortversion}
Because of space constraints, we are unable to provide full details for some of our algorithms and
lemma proofs in this version of the paper.
We refer the reader to the full version of the paper available at
\url{https://utdallas.edu/~kyle.fox/publications/holiest.pdf} for these details.
\end{shortversion}

\subsection{Additional Related Work}
\label{subsec:related}

Although they may sometimes go by different names such as \emph{uppermost} or \emph{rightmost}, the
idea of computing leftmost paths and flows in planar graphs appears as far back as the original
maximum flow-minimum cut paper of Ford and Fulkerson~\cite{ff-mfn-56}.
Several researchers have designed efficient algorithms for specializations of the maximum flow problem
in planar graphs using this idea~\cite{bg-pgtn-65,is-mfpn-79,h-mfpn-81,rww-vmppg-97,w-edstp-97,w-mstfp-97,bk-amfdp-09}.
There is a deep connection between flows in planar graphs and shortest paths in their duals
(see, for example, Venkatesan~\cite{v-anf-83}).
As far as we are aware, though, Klein~\cite{k-msspp-05} was the first to apply the idea of directly
computing \underline{\hspace{0.75cm}}most shortest path trees.

Khuller, Naor, and Klein~\cite{knk-lsfpg-93} observed that the set of integral \emph{circulations}
in a planar graph form a distributive lattice, and solutions to the minimum cost circulation problem
form a sublattice.
%\begin{longversion}
Indeed, a planar circulation is the boundary of a potential function (or \(2\)-chain) on the faces,
and the meet and join can be defined by taking the component-wise max and min of the potential
function, respectively.
%\end{longversion}
Many of the flow algorithms mentioned above actually find the top or bottom element in the
(sub)lattice.
Depending on which specifics one chooses, our lexicographic perturbation scheme simply enforces that
one choose the minimum flow or circulation according to this sublattice.
Matuschke and Peis~\cite{mp-lmfap-10b} show that the left/right relation on \(s,t\)-paths in planar
graphs also forms a lattice.

Bourke, Tewari, and Vinodchandran~\cite{btv-dprul-09} observed that the \emph{reachability} problem
for planar directed graphs lies in \emph{unambiguous log-space} (UL).
A key aspect of their algorithm is computing a set of lengths for the edges of a grid graph so that
shortest paths are unique.
Their edge weighting scheme was later extended to arbitrary planar graphs by Tewari and
Vinodchandran~\cite{tv-gtipg-12} and graphs embedded on constant genus surfaces by
Datta~\etal~\cite{dktv-scpmb-12}.
This latter result is similar to ours in that the length of each edge is the linear combination of
\(O(g)\) separate length functions including parts that encode the topology of paths and one part
encoding face containment for topologically trivial cycles.
However, our lexicographic perturbation scheme is arguably easier to implement than Datta~\etal's
scheme in that they (and Tewari and Vinodchandran~\cite{tv-gtipg-12}) must compute a straight-line
embedding of a subgraph of the input, while we work directly with the graph's \emph{combinatorial
embedding}.
Also, they use about twice as many length functions as we use vector components, and it is unclear
if their scheme is as directly useful as ours for designing a linear time algorithm for
multiple-source shortest paths in embedded graphs with small integer edge lengths.

\section{Preliminaries}
\label{sec:prelims}

We begin with an introduction to surface embedded graphs.
For more background we refer the reader to books and
surveys~\cite{h-at-02,m-t-00,eh-cti-10,z-tc-05,c-tags-12,mt-gs-01} related to the topic.

\paragraph*{Surfaces.}
A \EMPH{surface} or \(2\)-manifold with boundary \(\Sigma\) is a compact Hausdorff space where every
point lies in an open neighborhood homeomorphic to either the Euclidean plane or the closed half
plane.
The points whose neighborhoods are homeomorphic to the closed half plane constitute the
\EMPH{boundary} of the surface.
Every component of the boundary is homeomorphic to the unit circle.
A \EMPH{cycle} in the surface \(\Sigma\) is a continuous function \(\cycle : S^1 \to \Sigma\) where
\(S^1\) is the unit circle.
Cycle \(\cycle\) is \EMPH{simple} if \(\cycle\) is injective.
A \EMPH{path} \(\walk\) in the surface \(\Sigma\) is a continuous function \(\walk : [0,1] \to
\Sigma\);
again, \(p\) is simple if it is injective.
A \EMPH{loop} is a path \(\walk\) such that \(\walk(0) = \walk(1)\);
in other words, it is a cycle with a designated base point.
The \EMPH{genus} of the surface \(\Sigma\), which we denote as \(g\), is the maximum number of
pairwise disjoint simple cycles \(\cycle_1, \dots, \cycle_g\) in \(\Sigma\) such that \(\Sigma
\setminus (\cycle_1 \cup \dots \cup \cycle_g)\) is connected.
Surface \(\Sigma\) is \EMPH{non-orientable} if any subset of \(\Sigma\) is homeomorphic to the
\Mobius~band.
Otherwise, \(\Sigma\) is \EMPH{orientable}.
Up to homeomorphism, a surface is characterized by its genus, the number of boundary components, and
whether or not it is orientable.
We directly work only with orientable surfaces in this paper.%
\footnote{Cabello~\etal~\cite{cce-msspe-13} describe a reduction for their multiple-source shortest
paths algorithm in graphs embedded in non-orientable surfaces to the same problem in graphs
embedded in orientable surfaces.
We can apply our perturbation scheme or linear-time algorithm \emph{after} applying their reduction
in order to extend at least some of our results to graphs embedded in non-orientable surfaces.}

Let \(\walk_1\) and \(\walk_2\) be two paths in \(\Sigma\).
Paths \(\walk_1\) and \(\walk_2\) \EMPH{cross} if no continuous infinitesimal perturbation makes
them disjoint.
Otherwise, we call them \EMPH{non-crossing}.
They are \EMPH{homotopic} if one can be continuously deformed into the other without changing their
endpoints.
More formally, there must exist a \EMPH{homotopy} between them, defined as a continuous map \(h :
[0,1] \times [0,1] \to \Sigma\) such that \(h(0, \cdot) = p\) and \(h(1, \cdot) = q\).
Homotopy defines an equivalence relation over the set of paths with any fixed pair of endpoints.
%\begin{longversion}
A cycle is \EMPH{contractible} if it is homotopic to a constant map, and a loop is contractible if
it is homotopic to its base point.
The concatenation of a path \(\walk\) and loop \(\gamma\) with common endpoint is homotopic to
\(\walk\) if and only if \(\gamma\) is contractible.
%\end{longversion}

\paragraph*{Graph embeddings.}
The \EMPH{surface embedding} of an \emph{undirected} graph \(G\) with vertex set \(V\) and edge set
\(E\) is a drawing of \(G\) on a surface \(\Sigma\) which maps vertices to distinct points on
\(\Sigma\) and edges to internally disjoint simple paths whose endpoints lie on their incident
vertices' points.
A \EMPH{face} of the embedding is a maximally connected subset of \(\Sigma\) that does not
intersect the image of \(G\).
An embedding is \EMPH{cellular} if every face is homeomorphic to an open disc.
In a cellular embedding, every boundary component is covered by the image of a cycle in \(G\).
Let \(F\) be the set of faces of a cellular embedding, and let \(b\) be the number of boundary
components.
By Euler's formula, \(|V| - |E| + |F| = 2 - 2g - b\).

To more easily support directed graphs, we will assume \(G\) is connected and that its embedding is
given as a \EMPH{rotation system}.
These embeddings are sometimes referred to as \EMPH{combinatorial embeddings} as well (see, for
example, Eisenstat and Klein~\cite{ek-lafms-13}).
Let \(\dartsof{E}\) denote a collection of ``directed edges'' we refer to as \EMPH{darts}.
Let \(rev: \dartsof{E} \to \dartsof{E}\) be an involution on the darts we refer to as their
\EMPH{reversals}.
Each \EMPH{edge} \(e\) is an orbit in the involution \(\rev\).
We refer to one dart in \(e\)'s orbit as the \EMPH{canonical dart} of \(e\) and denote it by
\(\dartof{e}\).
In addition to \(\rev\), we have a permutation \(\permutation : \dartsof{E} \to \dartsof{E}\).
Each orbit of \(\permutation\) gives the counterclockwise cyclic ordering of darts ``directed into''
a vertex \(v\).
We refer to \(v\) as the \EMPH{head} of the darts in \(v\)'s orbit.
Vertex \(v\) is the \EMPH{tail} of these darts' reversals.
Orbits of the permutation \(\rev \circ \permutation\) give the \emph{clockwise} ordering of darts
around each face of the embedding.
We use the notation \(G = (V, E, F)\) to denote a surface embedded graph \(G\) with vertex
set \(V\), edge set \(E\), and face set \(F\).
From here on, we refer to such triples simply as \EMPH{graphs}.
In this setting, we can actually define the \EMPH{genus} \(g\) of \(G\) to be the \emph{solution} to
\(|V| - |E| + |F| = 2 - 2g\).

Given a graph \(G = (V, E, F)\), we define the \EMPH{dual graph} \(G^* = (F, E, V)\).
The given graph \(G\) is sometimes called the \EMPH{primal} graph.
Graph $G^*$ contains a vertex for every face of $G$, an edge for every edge of $G$, and a vertex for
every face of $G$.
Two dual vertices are connected by a dual edge if and only if the corresponding primal faces are separated by the corresponding primal edge.
In terms of combinatorial embeddings, the vertices of the dual graph are the orbits of the
permutation \(\rev \circ \permutation\);
moreover, the orbits of \(\rev \circ \permutation\) define the cyclic order of darts directed into
each dual vertex.
The drawing of dart \(d\) in the dual graph goes left to right across the drawing of \(d\) in the
primal graph.

\begin{figure}[t]
\centering
\begin{longversion}
\includegraphics[scale=0.3]{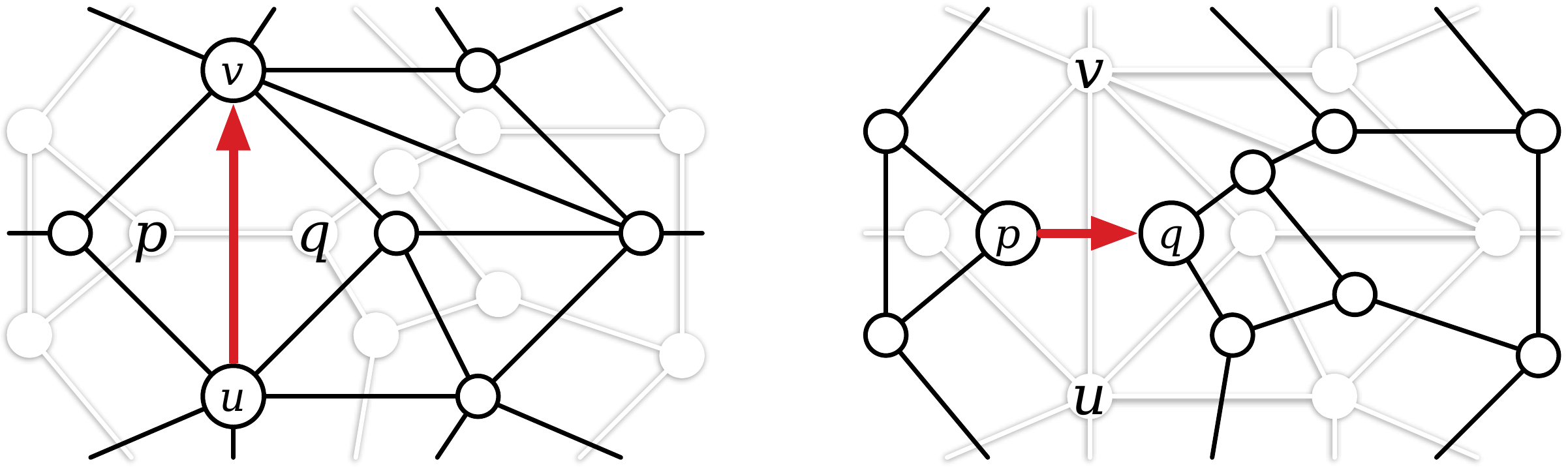}
\end{longversion}
\begin{shortversion}
\includegraphics[width=\columnwidth]{Fig/primal-dual}
\end{shortversion}
\caption{A dart \(u \arcto v\) in the primal graph $G$ and the corresponding dart  \(p \fencewith q\) in the dual graph $G^*$.}
\label{fig:dual}
\vspace{-5mm}
\end{figure}

For notational convenience, we will not distinguish between primal faces and dual vertices, primal
and dual darts/edges, or dual vertices and primal faces.
However, we will generally use the variables \(u\), \(v\), \(w\), \(x\), and \(y\) to denote primal vertices/dual
faces, and the variables \(o\), \(p\), \(q\), and \(r\) to denote dual vertices/primal faces.
We let \(u \arcto v\) denote a dart with tail \(u\) and head \(v\) in the primal graph, and let \(p
\fencewith q\) denote a dart with tail \(p\) and head \(q\) in the dual graph.
Finally, \(uv\) and \(p \wallwith q\) denote edges between vertices \(u\) and \(v\) or between dual
vertices \(p\) and \(q\), respectively.
%Their canonical darts are \(u \arcto v\) and \(p \fencewith q\), respectively.
See Figure~\ref{fig:dual}.

\paragraph*{Flows, homology, and final definitions.}
Flows are naturally defined either as \emph{non-negative} functions on the darts (without loss of
generality equal to zero on at least one dart of each edge) or as \emph{anti\-symmetric} functions
on the darts (where the values on the two darts of each edge sum to zero).
It will prove convenient to use the non-negative formulation to describe flows in the primal graph $G$ and the antisymmetric formulation to describe flows in the dual graph $G^*$.%
\footnote{This apparent asymmetry is actually a consequence of LP duality.
If we formulate minimum-cost flows in $G$ as a linear program using one formulation, the dual LP describes minimum-cost flows in $G^*$ in the other formulation!}
For convenience in our proofs, our formal definitions will require non-negativity only when
determining \emph{feasibility} of flows.

A (primal) \EMPH{flow} \(\flow : \dartsof{E} \to \R\) is an assignment of real values to the darts
of \(G\).
The \EMPH{imbalance} \(\imbalance \flow : V \to \R\) of flow \(\flow\) is the net flow going into
each vertex.
Formally, \(\imbalance \flow(v) = \sum_{u \arcto v} \flow(u \arcto v) - \sum_{v \arcto w} \flow(v
\arcto w)\).
%We can extend the definition to a subset \(V' \subseteq V\) of vertices by defining \(\imbalance
%\flow(V') = \sum_{v \in V'} \imbalance \flow(v)\).
Flow \(\flow\) is a \EMPH{circulation} if \(\imbalance(v) = 0\) for all \(v \in V\).
A \EMPH{potential function} \(\potential : F \to \R\) is an assignment of real values to
\emph{faces} of \(G\).
We say flow \(\flow\) is a \EMPH{boundary flow} of potential function \(\potential\) if for every
dart \(u \arcto v = q \fencewith p\), we have \(\flow(u \arcto v) - \flow(v \arcto u) =
\potential(p) - \potential(q)\).
In other words, high potentials to the \emph{right} of darts encourage high flow values while high
potentials to the left encourage low flow values.
All boundary flows are circulations.
Those familiar with concepts from algebraic topology may recognize the similarity between flows,
imbalances, and potentials functions with \(1\)-chains, boundaries of \(1\)-chains, and
\(2\)-chains, respectively.%
\footnote{This similarity is somewhat more natural with the antisymmetric formulation of flows.}
Two flows \(\flow_1\) and \(\flow_2\) are \EMPH{homologous} if their componentwise difference is the
boundary of some potential function.
Similar to homotopy, homology defines an equivalence relation over any set of flows with identical
vertex imbalances that is isomorphic with \(\R^{2g}\).

A \EMPH{dual flow} \(\coflow : \dartsof{E} \to \R\) assigns real values to the darts of the
\emph{dual} graph \(G^*\) such that \({\coflow(d) = -\coflow(\rev(d))}\) for every dart \(d\).  Equivalently, we consider a dual flow to be a function on the \emph{edges} of~$G^*$ by defining $\coflow(e) = \coflow(\dartof{e})$.
The \EMPH{dual imbalance} \(\coimbalance \coflow : F \to \R\) of dual flow \(\coflow\) is the total
dual flow going clockwise around each primal face, or equivalently, into each dual vertex.
Formally, \(\coimbalance \coflow(p) = \sum_{q \fencewith p} \coflow(q \fencewith p)\).
Let \(F' \subseteq F\) be any subset of faces.
Somewhat abusing notation, we let \(\cwneighbors(F')\) denote the \EMPH{clockwise neighborhood} of
\(F'\) so that \(\cwneighbors(F') = \set{p \fencewith q : p \notin F', q \in F'}\).
Thus, darts of \(\cwneighbors(F')\) are directed clockwise around the boundary of \(F’\) in the primal graph $G$ and enter \(F'\) in the dual graph $G^*$.

Let \(\cost : \dartsof{E} \to \R\) be a \EMPH{dart cost function}, let \(\capacity : \dartsof{E} \to
\R^+\) be a \EMPH{dart capacity function}, and let \(\demand : V \to \R\) be a \EMPH{vertex demand
function}.
The \EMPH{cost} of a flow \(\flow\) is \(\cost(\flow) = \sum_{d \in \dartsof{E}} \flow(d) \cdot
\cost(d)\).
A flow \(\flow\) is feasible with respect to \(\capacity\) and \(\demand\) if for all darts \(d \in
\dartsof{E}\) we have \(0 \leq \flow(d) \leq \capacity(d)\) and for all vertices \(v \in V\) we have
\(\imbalance \flow(v) = \demand(v)\).
A \EMPH{minimum-cost flow} with respect to \(\cost\), \(\capacity\), and \(\demand\) is a feasible
flow of minimum cost (if it exists).

\begin{longversion}
A (directed) \EMPH{path} \(\walk\) in \(G\) is a sequence of darts \(\seq{v_0 \arcto v_1, v_1 \arcto
v_2, \dots, v_{k-1} \arcto v_k}\) where consecutive darts share vertices.
\end{longversion}
\begin{shortversion}
A (directed) \EMPH{path} \(\walk\) in \(G\) is a sequence of darts \(\langle v_0 \arcto v_1, v_1
\arcto v_2,\)\linebreak
\(\dots, v_{k-1} \arcto v_k\rangle\) where consecutive darts share vertices.
\end{shortversion}
We often abuse terminology and identify a path with its drawing in \(G\)'s embedding.
A path is \EMPH{simple} if it does not repeat any vertices, except possibly its first and last
vertex.
%\begin{longversion}
The \EMPH{concatenation} of paths \(\walk_1\) and \(\walk_2\) is denoted \(\walk_1 \concat \walk_2\).
%\end{longversion}
Path \(\walk\) is a \EMPH{cycle} if \(v_0 = v_k\).
Abusing notation, we may treat \(\walk\) as a flow where \(\walk(d)\) is equal to the number of
times dart \(d\) appears in \(\walk\).
Given a vertex \(s \in V\), let \(\capacity : \dartsof{E} \to \R^+\) be a capacity function where
\(\capacity(d) = \infty\) for all \(d \in \dartsof{E}\), and let \(\demand : V \to \R\) be a demand
function where \(\demand(v) = 1\) for all \(v \neq s\) and \(\demand(s) = 1 - |V|\).
Given dart costs \(\cost : \dartsof{E} \to \R\) where no cycle has negative cost, the \EMPH{shortest
paths} from \(s\) to all other vertices can be defined as the set of paths starting at \(s\) and
composing the minimum-cost flow with respect to \(\capacity\) and \(\demand\).
Let \(\dist_c(s, t)\) denote the distance from \(s\) to \(t\) according to costs \(\cost\).
%\begin{longversion}
Let \(\shortestpath(s,t)\) denote the shortest path from \(s\) to \(t\).
%\end{longversion}

A \EMPH{spanning tree} \(\tree\) of \(G\) is a subset of edges that form a tree containing every
vertex.
We may \EMPH{root} \(\tree\) at a vertex \(s\) by considering the darts of \(\tree\) oriented
\emph{away} from \(s\).
Given a root \(s\) and vertex \(v \neq s\), the \EMPH{predecessor} of \(v\) in \(\tree\) is the
unique dart \(u \arcto v\) that lies on the path from \(s\) to \(v\) in \(\tree\).
%Vertex \(w\) is a descendant of \(v\) if \(v\) is on the path from \(s\) to \(w\) in \(\tree\).
Given an edge \(e \notin T\), the \EMPH{fundamental cycle} of \(e\) with \(T\), denoted
\(\fundcycle(T, e)\) is the unique simple cycle of edges in \(T + e\).
If \(e \in T\), then \(\fundcycle(T, e)\) is empty.
Given a dart \(d\) of an edge \(e\), its fundamental cycle \(\fundcycle(T, d)\) with \(T\) is the
orientation of \(\fundcycle(T, e)\) that contains \(d\).
A \EMPH{spanning cotree} \(\cotree\) of \(G\) is a subset of edges that form a spanning tree in the
dual graph.
We may root \(\cotree\) at a dual vertex \(r\) by now considering the darts of \(\cotree\) oriented
\emph{toward} \(r\).
The \EMPH{successor} of dual vertex \(p \neq r\) is the dart \(p \fencewith q\) that lies on the
dual path from \(p\) to \(r\) in \(\cotree\).
Dual vertex \(o\) is a \EMPH{descendant} of \(p\) in \(\cotree\) if \(p\) is on the dual path from
\(o\) to \(r\) in \(\cotree\).
\begin{longversion}
A \EMPH{tree-cotree decomposition}~\cite{e-dgteg-03} of \(G\) is a partition of \(E\) into \(3\)
disjoint edge subsets \(\tree, L, \cotree\), where \(\tree\) is a spanning tree of \(G\),
\(\cotree\) is a spanning cotree, and \(L\) is a set of \(2g\) leftover edges.
\end{longversion}

Given a vector \(a\), let \(a_i\) denote the \(i\)th component of \(a\).
\begin{longversion}
Let \((\tree, L, \cotree)\) be an arbitrary tree-cotree decomposition of \(G\).
Let \(\cycles = \seq{\cycle_1, \cycle_2, \dots, \cycle_{2g}}\) be some ordering of the \(2g\) dual
fundamental cycles of edges in \(L\) with \(\cotree\), and orient each \(\cycle_i\) in an arbitrary
direction.
The \EMPH{homology signature} \([e]\) of an edge \(e \in E\) with respect to \(\cycles\) is a
\(2g\)-dimensional integer vector where \([e]_i = 1\) if \(\dartof{e} \in \cycle_i\),
\([e]_i = -1\) if \(\rev(\dartof{e}) \in \cycle_i\), and \([e] = 0\) otherwise.
We can compute homology signatures in~\(O(gn)\) time by computing~\(\tree\) and~\(\cotree\)
in~\(O(n)\) time and then taking an~\(O(n)\)-time walk around each of the~\(O(g)\)
cycles~\(\cycle_i\), updating the \(i\)th component of each edge's signature vector along the walk.
\end{longversion}
\begin{shortversion}
In \(O(gn)\) time, we can compute a \EMPH{homology signature} \([e]\) for each edge \(e \in E\).
\end{shortversion}
Given a dart \(d\) of edge \(e\), we define the homology signature of \(d\) so that \([d] = [e]\) if
\(d = \dartof{e}\) and \([d] = -[e]\) otherwise.
Homology signatures given an implicit representation of a \emph{cohomology basis} in \(G\).
See Erickson and Whittlesey~\cite{ew-gohhg-05} and subsequest
papers~\cite{bccdw-ashba-12,en-mcsnc-11,cfn-csmcg-14,bcfn-mchbs-17}.
The homology signature of a flow \(\flow\) is \([\flow] = \sum_{d \in \dartsof{E}} \flow(d) \cdot
[d]\).
Finally, we have the following lemma, easily derived by modifying known results for homology
signatures.
\begin{lemma}
  \label{lem:signatures}
  Let \(\flow_1\) and \(\flow_2\) be two flows.
  Flow \(\flow_i\) is the boundary of some potential function if and only if \([\flow_i] = 0\).
  Further, \(\flow_1\) and \(\flow_2\) are homologous if and only if \([\flow_1] = [\flow_2]\).
\end{lemma}
In particular, Lemma~\ref{lem:signatures} implies that classes of flows with equivalent homology
signatures do not depend upon the particular choice of basis used to define the signatures.

\section{Holiest Perturbation}
\label{sec:scheme}

Let~\(G = (V,E,F)\) be a graph of size \(n\) and genus~\(g\).
Our lexicographic perturbation scheme relies on the properties of certain dual flows we refer to as
drainages.
Given a designated face \(r \in G\), we define a \EMPH{drainage} as a dual flow~\(\coflow\) where
\(\coimbalance{\coflow}(q) < 0\) for all \(q \in F \setminus \{r\}\).
The definition of a drainage immediately implies \(\coimbalance{\coflow}(r) = -\sum_{q \in F
\setminus \{r\}} \coimbalance{\coflow}(q) \).
Because \(r\) is the only face with positive dual imbalance with respect to \(\coflow\), we refer to
\(r\) as the \EMPH{sink} of the drainage \(\coflow\).
%We can easily extend the definition of a drainage to embedded graphs with multiple components
%as well by calling a flow a drainage if its restriction to each connected component is itself
%a drainage.
%Such drainages have one sink per graph component.

We now describe our perturbation scheme.
Let \(c : \dartsof{E} \to \R\) be a dart cost function.
We compute a set of homology signatures for the darts in~\(O(gn)\) time as described in
Section~\ref{sec:prelims}.
We then compute a drainage \(\coflow\) of \(G^*\) in \(O(n)\) time.
While any drainage will do, we describe one here that is easily computed.
We begin by computing an arbitrary spanning tree~\(\cotree\) of \(G^*\).
Let \(r \in F\) be an arbitrary face.
%, and imagine rooting \(\cotree\) at \(r\).
We define \(\coflow\) as if each face \(p \in F \setminus \{r\}\) is sending one unit of dual flow
along \(\cotree\) to \(r\).

Formally, we set the dual flow for each dart \(p \fencewith q \in \dartsof{E}\) as follows.
We root cotree \(\cotree\) at \(r\).
If \((p \fencewith q)\)'s edge is not in \(\cotree\), then \(\coflow(p \fencewith q) = 0\).
Otherwise, if \(p \fencewith q\) is the successor of \(p\) in \(\cotree\), then
\(\coflow(p \fencewith q)\) is the number of descendants of \(p\) (including \(p\) itself) in \(\cotree\);
otherwise, \(\coflow(p \fencewith q)\) is the negation of the number of descendants of \(q\) in \(\cotree\).
One can easily verify that \(\coimbalance{\coflow}(q) = -1\) for all \(q \in F \setminus \{r\}\) and
\(\coimbalance{\coflow}(r) = \numfaces - 1\).

We now redefine the costs of darts in \(G\).
Intuitively, we add a sequence of progressively smaller infinitesimal values to the cost of each
dart based partially on the homology signatures of their edges and the dual flow they carry from the
drainage \(\coflow\).
More concretely, we define a new dart cost function \(\cost' : \dartsof{E} \to \R \times \N^{2g+2}\) as follows.
Let \(\tuplecat\) denote the concatenation of two vectors, and define
\[
  \cost’(d) := (\cost(d), 1) \tuplecat \homsig{d} \tuplecat (\coflow(d)).
  %\text{
%  and}\\
%  \cost'(\rev(\dartof{e})) &:= (\cost(\rev(\dartof{e})), 1) \tuplecat -\homsig{e} \tuplecat
%  (-\coflow(e)).
\]

The definition for the cost of a flow \(\flow\) can be modified easily to work with our new cost
function: \(\cost'(\flow)\) is the vector \(\sum_{d \in \dartsof{E}} \flow(d) \cdot \cost’(d)\).
Given a cost vector \(\cost'\) for either a single dart or a whole flow, we refer to the components of dart and flow costs determined by homology signatures as the
\EMPH{homology parts} of \(\cost'\), denoted \(\homsig{\cost'}\).
The last component is referred to as the \EMPH{face part} and denoted \(\coflow(\cost')\).
Comparisons between dart and flow costs are performed lexicographically.
As a consequence, any minimum-cost flow with respect to \(\cost'\) is also a minimum-cost flow with
respect to the original scalar cost function \(\cost\).
Intuitively, minimizing the cost of a flow according to \(\cost'\) means first minimizing the
original cost according to \(\cost\), then minimizing the sum of the darts' flow values, then
lexicographically minimizing the homology class of the flow, and finally choosing the leftmost flow
subject to all other conditions.
In particular, when \(g = 0\), one is computing a leftmost minimum-cost flow (after also minimizing
the sum of darts' flows).
The following lemma is immediate.
\begin{lemma}
  \label{lem:cost_components}
  For any flow \(\flow\) and \(j \in \{1,\ldots,2g\}\), we have \(\homsig{\flow}_j =
  \cost'_{j+2}(\flow)\).
  In addition,
\begin{longversion}
  \[
  \cost'_{2g+3}(\flow) = \sum_{d \in \dartsof{E}} \flow(d) \cdot \coflow(d).
  \]
\end{longversion}
\begin{shortversion}
  \(
  \cost'_{2g+3}(\flow) = \sum_{d \in \dartsof{E}} \flow(d) \cdot \coflow(d)
  \).
\end{shortversion}
\end{lemma}

Computing the perturbations takes \(O(gn)\) time total.
The time for every addition, multiplication, and comparison is now \(O(g)\) instead of \(O(1)\).
For planar graphs in particular, this scheme requires only linear preprocesing time, and
combinatorial algorithms relying on the new costs do not have higher asymptotic running times.
Note that the cost of each dart \(d\) is strictly larger than \((\cost(d), 0, \ldots, 0)\).
No negative-cost directed cycles are created, even when some directed cycles had length~\(0\)
originally, meaning shortest paths are still well-defined.
In fact, the perturbation scheme does not create any new negative-cost darts, so combinatorial
shortest path algorithms relying on non-negative dart costs still function correctly.
As stated, however, these algorithms and those for negative costs do slow down by a factor of
\(g\).

\subsection{Analysis}
\label{subsec:scheme-analysis}
We now prove our perturbation scheme guarantees uniqueness of minimum-cost flows and shortest
paths as promised.
We begin by discussing the former as the latter follows as an easy consequence.
The key observation behind our proof is that drainages encode the total imbalance of vertices
lying on one side of a dual cut.
In turn, we use this observation to show that every non-trivial circulation has a part of non-zero
cost after using our perturbation scheme.
The former observation is a slight generalization of one by Patel~\cite[Lemma 2.4]{p-deeoc-13}
who in turn generalized a result for planar graphs by Park and Phillips~\cite{pp-fmcpg-93}.
\begin{longversion}
We could use Patel's result directly for the tree-based drainage described above.
However, we are able to give a short proof of the more general lemma below.
\end{longversion}
\begin{lemma}
  \label{lem:drainages}
  Let \(\coflow\) be a drainage with sink \(r\), and let \(F' \subseteq F\).
  We have
  \[
  \sum_{d \in \cwneighbors(F')} \coflow(d) = \sum_{q \in F'}
  \coimbalance{\coflow}(q) = -\sum_{q \in F \setminus F'} \coimbalance{\coflow}(q).
  \]
  In particular, for non-empty \(F' \neq F\), we have \(\sum_{d \in \cwneighbors(F')} \coflow(d)\)
  is positive if \(r \in F'\) and negative otherwise.
\end{lemma}
\begin{longversion}
\begin{proof}
  We have
  \begin{align*}
    \sum_{d \in \cwneighbors(F')} \coflow(d)
    &= \sum_{p \fencewith q \in \dartof{E} : p \notin F', q \in F'}\coflow(p \fencewith q) + \sum_{p \wallwith q \in E :
    p,q \in F'}[\coflow(p \wallwith q) - \coflow(p \wallwith q)]\\
    &= \sum_{p \fencewith q \in \dartof{E} : p \notin F', q \in F'}\coflow(p \fencewith q) + \sum_{p \wallwith q \in E :
    p,q \in F'}[\coflow(p \fencewith q) + \coflow(q \fencewith p)]\\
    &= \sum_{p \fencewith q \in \dartsof{E} : q \in F'} \coflow(p \fencewith q)\\
    &= \sum_{q \in F'} \sum_{p \fencewith q \in \dartsof{E}} \coflow(p \fencewith q)\\
    &= \sum_{q \in F'} \coimbalance{\coflow}(q).
  \end{align*}
  The final lemma statement follows easily from the definition of drainages.
\end{proof}
\end{longversion}

Fix a cost function \(\cost : \dartsof{E} \to \R\), and let \(\cost'\) be the perturbation of
\(\cost\) defined above.
Also fix a capacity function \(\capacity : \dartsof{E} \to \R^+\) and a demand function \(\demand :
V \to \R\).
%, and a set \(\boundaries \subseteq F\) of boundary faces.
We give the following lemma, which immediately implies the uniqueness of minimum-cost flows.
%The latter parts will prove useful in the next section in which we discuss extensions to our main
%uniqueness result.
\begin{lemma}
  \label{lem:something_cheaper}
  Let \(\flow_1\) and \(\flow_2\) be distinct feasible flows with respect to \(\capacity\) and
  \(\demand\).
  There exists a feasible flow~\(\flow\) such that at least one of \(\cost'(\flow) <
  \cost'(\flow_1)\) or \(\cost'(\flow) < \cost'(\flow_2)\) is true.
%  with the following properties.
%  \begin{enumerate}
%    \item
%      At least one of \(\cost'(\flow) < \cost'(\flow_1)\) or \(\cost'(\flow) < \cost'(\flow_2)\)
%      is true.
%    \item
%      For each edge \(e \in E\), if \(\flow_1(e) = \flow_2(e)\), then \(\flow(e) = 0\).
%    \item
%      If \(\flow_1\) and \(\flow_2\) are homologous with respect to \(\boundaries\), then
%      \(\flow\) is homologous to both as well.
%  \end{enumerate}
\end{lemma}
We emphasize that our perturbation scheme does not guarantee \emph{all} feasible flows have distinct
costs, and it may be that \(\cost'(\flow_1) = \cost'(\flow_2)\).
However, we would then have \(\flow\) costing strictly less than both \(\flow_1\) and \(\flow_2\),
implying neither \(\flow_1\) nor \(\flow_2\) is a minimum-cost feasible flow.
%\begin{shortversion}
%\note{Kyle: STOC reviewer A and I would like really like this proof to appear in the STOC
%camera-ready.}
%In the proof of Lemma~\ref{lem:something_cheaper} (in the full version), we define \(\tilde{\flow}
%= \flow_2 - \flow_1\).
%If \(\tilde{\flow}\) has non-zero homology signature, then Lemma~\ref{lem:cost_components} implies
%it has non-zero cost.
%Otherwise, we can extract a non-trival flow \(\flow_\eps\) from \(\tilde{\flow}\) that encloses a
%non-empty set of faces~\(F'\).
%From Lemma~\ref{lem:drainages}, we see \(\flow_\eps\) has non-zero cost.
%Either way, we have found a flow of non-zero cost that we can add to or subtract from one of
%\(\flow_1\) or \(\flow_2\) to create a cheaper flow.
%\end{shortversion}
%\begin{longversion}
\begin{proof}
  We will prove existence of a circulation \(\hat{\flow}\) such that \(\flow_i + \hat{\flow}\) is
  feasible for some \(i \in \{1,2\}\) and \(\cost'(\hat{\flow}) < 0\).
  %For each edge \(e \in E\), if \(\flow_1(e) = \flow_2(e)\), then \(\hat{\flow}(e) = 0\).
  %If \(\flow_1\) and \(\flow_2\) are homologous with respect to \(\boundaries\), then
  %\(\hat{\flow}\) is a boundary circular with respect to \(\boundaries\).
  We set \(\flow = \flow_i + \hat{\flow}\), proving the lemma.

  Let \(\tilde{\flow} = \flow_2 - \flow_1\).
  Both \(\flow_2\) and \(\flow_1\) are feasible with respect to demand function \(\demand\), so
  \(\tilde{\flow}\) must be a non-trivial circulation.
  Further, for any dart~\(d\) and scalar \(a\) with \(0 \leq a \leq 1\), we have
  \begin{align*}
    \flow_1(d) + a\tilde{\flow}(d)
    &= \flow_1(d) + a(\flow_2(d) - \flow_1(d))\\
    &\geq \min\{\flow_1(d), \flow_1(d) + \flow_2(d) - \flow_1(d)\}\\
    &\geq 0 \text{ and}\\
    \flow_1(d) + a\tilde{\flow}(d) &= \flow_1(d) + a(\flow_2(d) - \flow_1(d))\\
    &\leq \max\{\flow_1(d), \flow_1(d) + \flow_2(d) - \flow_1(d)\}\\
    &\leq \capacity(d).
  \end{align*}

  In other words, we can add any circulation consisting of scaled down components of \(\tilde{\flow}\)
  to \(\flow_1\) and still have a feasible flow.
  Similarly, we can add any circulation consisting of scaled down components of \(-\tilde{\flow}\)
  to \(\flow_2\) and still have a feasible flow.
  We now consider two cases.
  \paragraph{Case 1: \([\tilde{\flow}] \neq 0\).}
  Let \([\tilde{\flow}]_j\) be non-zero.
  Lemma~\ref{lem:cost_components} implies \(\cost'_{j+2}(\tilde{\flow})\) is also non-zero, further
  implying \(\cost'(\tilde{\flow})\) is itself non-zero.
  If \(\cost'(\tilde{\flow}) < 0\), then let \(\hat{\flow} = \tilde{\flow}\).
  Otherwise, let \(\hat{\flow} = -\tilde{\flow}\).
  %Note that if \(\flow_1\) and \(\flow_2\) are homologous with respect to \(\boundaries\), then
  %\([\flow_1] = [\flow_2]\), and this case does not apply.

  \paragraph{Case 2: \([\tilde{\flow}] = 0\).}
  \def\plower{{\underline{\potential}}}
  \def\pupper{{\overline{\potential}}}

  We consider two subcases.

  First, suppose there exists an edge \(e \in E\) such that \(\tilde{\flow}(\dartof{e}) =
  \tilde{\flow}(\rev(\dartof{e})) > 0\).
  Let \(\flow_e : \dartsof{E} \to \R\) be a flow that is everywhere-zero except
  \(\flow_e(\dartof{e}) = \flow_e(\rev(\dartof{e})) = \tilde{\flow}(\dartof{e})\).
  Then, \(\cost'_{2}(\flow_e) = 2\tilde{\flow}(\dartof{e})\), implying \(\cost'(\flow_e)\) is
  non-zero.
  If \(\cost'(\flow_e) < 0\), then let \(\hat{\flow} = \flow_e\).
  Otherwise, let \(\hat{\flow} = -\flow_e\).

  Now, suppose there is no such edge \(e\) as defined above.
  Then, Lemma~\ref{lem:signatures} implies \(\tilde{\flow}\) is a boundary flow for some non-trivial
  potential function~\(\potential\).
  %If \(\flow_1\) and \(\flow_2\) are homologous with respect to \(\boundaries\), then we can
  %assume \(\potential(\boundary) = 0\) for all \(\boundary \in \boundaries\).
  Let \(\plower = \min_{q \in F} \potential(q)\) and \(\pupper = \max_{q \in F} \potential(q)\).
  Because \(\potential\) is non-trivial, at least one of \(\plower\) and \(\pupper\) is non-zero.
  Assume \(\pupper \neq 0\); the other case is similar.
  Let \(F_\pupper = \{q \in F : \potential(q) = \pupper\}\), and let \(\dartsof{E}_\pupper =
  \{\fence{p}{q} \in \dartsof{E} : p \in F \setminus F_\pupper, q \in F_\pupper\}\).
  For each dart \(\fence{p}{q} \in \dartsof{E}_\pupper\), we have \(\tilde{\flow}(\fence{p}{q}) -
  \tilde{\flow}(\fence{q}{p}) > 0\), because \(\potential(q) > \potential(p)\).

  Let \(\eps = \min_{d \in \dartsof{E}_\pupper} (\tilde{\flow}(d) - \tilde{\flow}(\rev(d)))\).
  For each dart \(d \in \dartsof{E}_\pupper\), let \(a_d = \eps / (\tilde{\flow}(d) -
  \tilde{\flow}(\rev(d)))\).
  Note that \(0 < a_d \leq 1\).
  Finally, let \(\flow_\eps : \dartsof{E} \to \R\) be a flow that is everywhere-zero except for each
  dart \(d \in E_\pupper\), we have \(\flow_\eps(d) = a_d \tilde{\flow}(d)\) and
  \(\flow_\eps(\rev(d)) = a_d \tilde{\flow}(\rev(d))\);
  in other words, \(\flow_\eps(d) - \flow_\eps(\rev(d)) = \eps\).
  %Finally, let \(\potential_\eps : F \to \R\) be the potential function such that
  %\(\potential_\eps(p) = \eps\) for all \(p \in F_\pupper\) and \(\potential_\eps\) is otherwise
  %zero.
  %If \(\flow_1\) and \(\flow_2\) are homologous with respect to \(\boundaries\), then \(F_\pupper
  %\subseteq F \setminus \boundaries\).
  %Therefore, \(\partial \potential_\eps\) is a boundary circulation with respect to
  %\(\boundaries\).

  %For each dart \(d \in \dartsof{E}_\pupper\), we have \(\partial \potential_\eps(d) \leq
  %\tilde{\flow}(d) \leq \residual{\capacity}{\flow_1}(d)\).
  %Circulation \(\partial \potential_\eps\) is otherwise non-positive, and
  %\(\residual{\capacity}{\flow_1}\) is non-negative, so \(\partial \potential_\eps(d)\) is
  %feasible for \(\residual{\capacity}{\flow_1}\).
  %Similarly, \(-\partial \potential_\eps(d)\) is feasible for \(\residual{\capacity}{\flow_2}\).
  %Also, \(\partial \potential_\eps(d)\) is non-zero only on edges for which \(\tilde{\flow}\) is
  %non-zero.
  Let~\(\coflow\) be the drainage used to define \(\cost'\).
  Lemmas~\ref{lem:cost_components} and~\ref{lem:drainages} imply \(\cost'_{2g+3}(\flow_\eps) =
  \eps\sum_{d \in \cwneighbors(F_\pupper)}\).
  Set \(F_\pupper\) is a non-empty strict subset of \(F\), because \(\tilde{\flow}\) is
  non-trivial.
  Therefore, Lemma~\ref{lem:drainages} also implies \(\cost'_{2g+3}(\flow_\eps)\) is
  non-zero, meaning \(\cost'(\flow_\eps)\) is also non-zero.
  If \(\cost'(\flow_\eps) < 0\), then let \(\hat{\flow} = \flow_\eps\).
  Otherwise, let \(\hat{\flow} = -\flow_\eps\).
\end{proof}
%\end{longversion}
\begin{theorem}
  \label{thm:unique_flow}
  Let \(G = (V, E, F)\) be a graph of genus~\(g\), let \(\cost : \dartsof{E} \to \R\) be a
  dart cost function, and let \(\cost' : \dartsof{E} \to \R\) be the output of our lexicographic
  perturbation scheme on \(\cost\).
  Let \(\capacity : \dartsof{E} \to \R^+\) and \(\demand : V \to \R\) be a dart capacity and vertex
  demand function, respectively.
  The minimum-cost feasible flow with respect to \(\cost'\), \(\capacity\), and \(\demand\) is
  unique and is a minimum-cost feasible flow with respect to \(\cost\) as well.
\end{theorem}

%\begin{longversion}
Recall, the shortest \(s,t\)-path problem is a special case of minimum-cost flow where for each dart
\(d\), \(\capacity(d) = \infty\).
In addition, all demands are zero except \(\demand(t) = -\demand(s) = 1\).
Every directed cycle has its cost strictly increase, so if \(\cost\) has no negative-length
directed cycles, then \(\cost'\) has no negative or even \emph{zero}-length directed cycles.
Any feasible flow with a directed cycle \(\cycle\) can be made cheaper by removing \(\cycle\).
The unique minimum-cost flow with respect to \(\cost'\) guaranteed by
Theorem~\ref{thm:unique_flow} is a directed path from \(s\) to \(t\).
%\end{longversion}
%\begin{shortversion}
%Recall, shortest paths are simply a special case of minimum-cost flow.
%\end{shortversion}
\begin{corollary}
  \label{cor:unique_path}
  Let \(G = (V, E, F)\) be a graph of genus~\(g\), let \(\cost : \dartsof{E} \to \R\) be a
  dart cost function, and let \(\cost' : \dartsof{E} \to \R\) be the output of our lexicographic
  perturbation scheme on \(\cost\).
  Let \(s,t \in V\).
  The shortest \(s,t\)-path with respect to \(\cost'\) is unique and is a shortest \(s,t\)-path
  with respect to \(\cost\) as well.
\end{corollary}

From here on, we refer to the unique minimum-cost flows and shortest paths guaranteed by our
perturbation scheme as \EMPH{homologically lexicographic least leftmost} or \EMPH{holiest} flows and
paths.

%\begin{shortversion}
%\paragraph*{Minimum cut in directed planar graphs.}
%As discussed in the introduction, our perturbation scheme can be used in a black box fashion to
%immediately derandomize the \(O(n \log \log n)\) time minimum cut algorithm of
%Mozes~\etal~\cite{mnnw-mcdpg-18} for directed planar graphs.
%The only change necessary to derandomize their algorithm is to guarantee uniqueness of shortest
%paths in the \emph{dual graph}.
%\end{shortversion}
%\begin{longversion}
\section{Minimum Cut in Directed Planar Graphs}
\label{sec:planar_cut}
As discussed in the introduction, our perturbation scheme can be used in a black box fashion to
immediately derandomize the \(O(n \log \log n)\) time minimum cut algorithm of
Mozes~\etal~\cite{mnnw-mcdpg-18} for directed planar graphs.
The only change necessary to derandomize their algorithm is to guarantee uniqueness of shortest
paths in the \emph{dual graph}.
\begin{corollary}
  \label{cor:global_cut}
  Let \(G = (V, E, F)\) be a planar graph of size \(n\), and let \(\cost : \dartsof{E} \to \R\) be a
  dart cost function.
  There exists a deterministic algorithm that computes a global minimum cut of \(G\) with respect to
  \(\cost\) in \(O(n \log \log n)\) time.
\end{corollary}
%\end{longversion}

\section{Multiple-Source Shortest Paths}
\label{sec:MSSP}

Our scheme can be used in a black box fashion in the multiple-source shortest paths algorithm of
Cabello~\etal~\cite{cce-msspe-13}.
However, they depend on another property of the dart costs beyond uniqueness of shortest paths.
%\begin{longversion}
Our perturbation scheme does guarantee the additional property, but we must first describe their
algorithm in order to even explain what that property is.
Understanding their algorithm is also a crucial first step in describing our linear-time algorithm
for embedded graphs with constant genus and small integer dart costs.
In order to more cleanly explain our linear-time algorithm in later sections, we describe a slight
variant of Cabello~\etal's algorithm.
This variant is based on the linear-time multiple-source shortest paths algorithm of Eisenstat and
Klein~\cite{ek-lafms-13} for planar graphs with small integer dart costs.
%\end{longversion}
%\begin{shortversion}
%We describe a variant of their algorithm proposed by Eisenstat and Klein~\cite{ek-lafms-13}
%\end{shortversion}

Let~\(G = (V,E,F)\) be an embedded graph of size~\(n\) and genus~\(g\), and let \(\cost : \dartsof{E} \to \R\)
be a cost function on the darts.
Let \(r \in F\) be an arbitrary face of \(G\) from whose vertices we want to preprocess shortest
paths with regard to \(\cost\).
The multiple-source shortest paths algorithm begins by computing a shortest path tree \(T\) rooted at
an arbitrary vertex of \(r\).
The algorithm proceeds by iteratively changing the source of the shortest path tree to each of the
vertices in order around \(r\);
each change is implemented as a sequence of \EMPH{pivots} wherein one dart \(x \arcto y\) enters
\(T\) and another dart \(w \arcto y\) leaves.

Consider one \EMPH{iteration} of the algorithm where the source moves from a vertex \(u\) to a
vertex \(v\).
To move the source, the algorithm performs a \EMPH{special pivot}.
Let \(x \arcto v\) be the predecessor dart of \(v\) in \(T\).
During the special pivot, the algorithm removes dart \(x \arcto v\) from \(T\) and adds dart \(v
\arcto u\);
afterward, \(T\) is rooted at \(v\).
Let \(\parameter \in \R\), and let \(\cost_\parameter : \dartsof{E} \to \R\) be a parameterized cost
function where \(\cost(v \arcto u) = \parameter\) and \(\cost_\parameter(d) = \cost(d)\) for all \(d
\neq v \arcto u\).
This special pivot is accompanied by temporarily redefining the dart costs in terms of
\(\cost_\parameter\) with \(\parameter\) initially set to \(-\dist_{\cost}(u, v)\).
Changing the costs in this way guarantees that \(T\) is a shortest path tree rooted at \(v\)
\emph{given dart \(v \arcto u\) has cost \(\parameter\)}.

Conceptually, the rest of the iteration is performed by continuously increasing \(\parameter\) until
it reaches \(\cost(v \arcto u\)), the original cost of \(v \arcto u\), and maintaining \(T\) as a
shortest path tree as \(\parameter\) is increased.
Following convention from Cabello~\etal~\cite{cce-msspe-13}, we say a vertex \(x\) is
\EMPH{\color{red} red} if the \(v\) to \(x\) path in \(T\) uses dart \(v \arcto u\);
otherwise, the vertex is \EMPH{\color{blue} blue}.
Let \(\dist_{\parameter}\) denote \(\dist_{\cost_\parameter}\) for simplicity.%
\begin{longversion}
Define the \EMPH{slack} of dart \(x \arcto y\) with regard to \(\parameter\) to be
\[
\slack_\parameter(x \arcto y) := \dist_{\parameter}(v, x) + \cost_\parameter(x \arcto y) -
\dist_\parameter(v, y).
\]
%The slack of dart \(x \arcto y\) is the cost of the directed fundamental cycle \(x \arcto y\) with
%\(T\).
\end{longversion}
\begin{shortversion}
Define the \EMPH{slack} of dart \(x \arcto y\) with regard to \(\parameter\) to be
\(
\slack_\parameter(x \arcto y) := \dist_{\parameter}(v, x) + \cost_\parameter(x \arcto y) -
\dist_\parameter(v, y) \geq 0
\).
\end{shortversion}
For any \(\parameter \in \R, x \arcto y \in \dartsof{E}\), we have \(\slack_\parameter(x \arcto y)
\geq 0\).
We say dart \(x \arcto y\) is \EMPH{tense} if \(\slack_\parameter(x \arcto y) = 0\).
A spanning tree \(T'\) rooted at \(v\) is a shortest path tree if and only if every dart in \(T'\)
is tense.
Dart \(x \arcto y\) is \EMPH{active} if \(\slack_\parameter(x \arcto y)\) is decreasing in
\(\parameter\).
%\begin{longversion}
A dart \(x \arcto y\) is active if and only if \(x\) is blue and \(y\) is red~\cite[Lemma
3.1]{cce-msspe-13}.
%\end{longversion}
All active darts see the same rate of slack decrease as \(\parameter\) rises.

As \(\parameter\) increases, it reaches certain critical values where an active dart \(x \arcto y\)
becomes tense.
The algorithm then performs a pivot by inserting \(x \arcto y\) into \(T\) and removing the original
predecessor \(w \arcto y\) of~\(y\).
Because \(x \arcto y\) is tense when the pivot occurs, \(T\) remains a shortest path tree rooted at
\(v\).
Note that, with the exception of \(v \arcto u\) during the special pivot, slacks do not change
during pivots.

\begin{longversion}
Using appropriate dynamic-tree data structures~\cite{st-dsdt-83,t-dtste-97,hk-rfdga-99,tw-satt-05},
these critical values for \(\parameter\) can be computed and pivots can be performed in amortized
\(O(\log n)\) time per pivot.
Eisenstat and Klein use simpler data structures for the case of planar graphs with small integer
costs;
see Section~\ref{sec:linear-time} for details.

Across all iterations, the total number of pivots is \(O(gn)\)~\cite[Lemma 4.3]{cce-msspe-13}.
Therefore, between performing pivots and some \(O(g n \log n)\) time additional work, the algorithm
of Cabello~\etal~spends \(O(g n \log n)\) time total.
However, their algorithm and analysis depend upon two genericity assumptions:
all vertex-to-vertex shortest paths are unique, and exactly one dart becomes tense at each critical
value of \(\parameter\).
\end{longversion}

\begin{shortversion}
The algorithm and analysis of Cabello~\etal~\cite{cce-msspe-13} depend upon two genericity
assumptions:
all shortest paths are unique, and exactly one dart becomes tense at each critical value of
\(\parameter\).
Suppose we apply our perturbation scheme and work perturbed costs~\(\cost'\).
Observe that~\(\parameter\) is now an increasing vector instead of a scalar.
Corollary~\ref{cor:unique_path} guarantees that the first assumption of Cabello~\etal\ is now
enforced.
The second assumption can be shown to hold by treating the selection of pivots as a minimum-cost
flow problem, for which our perturbation scheme guarantees a unique solution
(see Cabello~\etal~\cite[Section 6]{cce-msspe-13}).
\end{shortversion}

\begin{longversion}
Suppose we apply our lexicographic perturbation scheme so we are maintaining the holiest shortest
path tree \(T\).
Let \(\cost' : \dartsof{E} \to \R \times \N^{2g+2}\) be the perturbed costs.
Observe that~\(\parameter\) is now an increasing vector instead of a scalar.
Corollary~\ref{cor:unique_path} guarantees that the first assumption is enforced.
For the second assumption, we prove the following lemma.
\begin{lemma}
  Consider an iteration where the source of \(T\) moves from vertex \(u\) to vertex \(v\).
  For each value of \(\parameter\), there is at most one active dart of minimum slack.
\end{lemma}
\begin{proof}
  Suppose there is at least one active dart.
  Let \(\capacity : \dartsof{E} \to \R^+\) be the dart capacities and \(\demand : V \to \R\) be the
  vertex demands for shortest paths rooted at \(v\), and let \(\flow\) be the minimum-cost flow with
  respect to \(\cost'_\parameter\), \(\capacity\), and \(\demand\) that uses darts of \(T\).
  Let \(\demand_p : V \to \R\) (`p' stands for pivot) be a vertex demand function that is
  zero-everywhere except \(\demand_p(u) = -\demand_p(v) = 1\), and let \(\cost_p : \dartsof{E} \to
  \R^+\) be a \emph{residual} cost function where \(\cost_p(v \arcto u) = \infty\), \(\cost_p(d) =
  -\cost'_\parameter(\rev(d))\) if \(\rev(d) \neq v \arcto u\) and \(\rev(d)\) lies on a shortest
  path according to \(\cost'_\parameter\), and \(\cost_p(d) = \cost'_\parameter(d)\) otherwise.
  Finding the active dart of minimum slack is equivalent to computing a holiest flow with respect to
  \(\cost_p\), \(\capacity\), and \(\demand_p\).
  However, the holiest flow is unique by Theorem~\ref{thm:unique_flow}.
  %A pivot adding an active dart \(x \arcto y\) to \(T\) is equivalent to adding some multiple of the
  %flow \(\flow' = (\shortestpath(v, x) \concat x \arcto y) - \shortestpath(u, y)\) to \(\flow\).
  %We have \(\slack_{\parameter}(x \arcto y) = \cost'_{\parameter}(\flow') + \cost'_{\parameter}(v
  %\arcto u)\).
  %Therefore, an active dart of minimum slack is also the active dart used by a minimum-cost flow
  %\(\flow''\) with respect to \(\cost'\), \(\capacity_{\flow}\), and \(\demand'\) \emph{subject to
  %\(\flow''(v \arcto u) = 0\)}.
%
  %Now, suppose for the sake of proof that we increase \(\cost(v \arcto u)\) to \(\infty\) and
  %redefine the first coordinate of \(\cost'(v \arcto u)\) accordingly.
  %Now, \(\flow''\) is a minimum-cost flow with respect to \(\cost'\), \(\capacity_{\flow}\), and
  %\(\demand'\) with no additional restrictions.
  %But, such minimum-costs flows are unique by Theorem~\ref{thm:unique_flow}.
  %Therefore, the active dart of minimum slack is also unique.
\end{proof}
\end{longversion}

After applying our perturbation scheme, the time to do basic operations on costs increases by a
factor of \(g\).
%In addition, storage requirements increase by a factor of~\(g\).
%the graph and dart costs now take \(O(gn)\) space total to store.
\begin{longversion}
In particular, we can construct the multiple-source shortest paths data structure of Cabello
\etal~\cite{cce-msspe-13} for \emph{unperturbed} dart costs with only a~\(g\) factor increase in
the construction time.
The~\(O(g n \log n)\) space needed to store the final data structure is already more than
the~\(O(gn)\) space we need to store perturbed edge lengths and auxiliary structures used during
construction.
\end{longversion}
\begin{shortversion}
We conclude:
\end{shortversion}
\begin{theorem}
  \label{thm:MSSP}
  Let \(G = (V, E, F)\) be a graph of size \(n\) and genus \(g\), let \(\cost : \dartsof{E} \to \R\)
  be a dart cost function, and let \(r \in F\) be any face of \(G\).
  We can deterministically preprocess \(G\) in \(O(g^2 n \log n)\) time and \(O(g n \log n)\) space
  so that the (unperturbed) length of the shortest path from any vertex incident to \(r\) to any
  other vertex can be retrieved in \(O(\log n)\) time.
\end{theorem}

\begin{longversion}
\section{The Leafmost Rule in Planar Graphs}
\label{sec:leafmost}

In the previous section, we discussed using our perturbation scheme to efficiently solve the
multiple-source shortest paths problem in embedded graphs.
Recall, perturbations (deterministic or randomized) are not required for algorithms designed to
compute multiple-source shortest paths exclusively in \emph{planar}
graphs~\cite{k-msspp-05,ek-lafms-13}.
In fact, the linear-time algorithm of Eisenstat and Klein~\cite{ek-lafms-13} \emph{cannot} rely on
perturbation schemes as it crucially depends upon all dart costs being small non-negative integers.
Instead, these algorithms rely on the \emph{leafmost} rule for selecting darts to pivot into the
shortest path tree.
The leafmost rule is also used in efficient \(s,t\)-maximum flow algorithms based on parametric
shortest paths in the dual graph~\cite{bk-amfdp-09,e-mfpsp-10,ek-lafms-13}.

Let~\(G = (V,E,F)\) be a connected planar graph of size~\(n\), and let \(\cost : \dartsof{E} \to
\R\) be a cost function on the darts.
Let \(T\) be the shortest path tree maintained while running Section~\ref{sec:MSSP}'s
multiple-source shortest paths tree algorithm along face \(r \in F\).
Suppose we are moving the source of \(T\) from vertex \(u\) to vertex \(v\), and let \(v \arcto u =
r \fencewith q\).
Let \(\cotree\) be a spanning tree of \(G^*\) complementary to \(T\) with darts oriented toward
\(r\).
In other words, the darts of \(T\) and \(\cotree\) belong to distinct edges.
The set of active darts are precisely the darts in the \(q,r\)-path through
\(\cotree\)~\cite{cce-msspe-13}.
%(declaring them to be the reversal darts is consistent with our
%convention that dual edges are turned \(90\degree\) relative to their primal counterparts).
In this setting, the leafmost rule says we should pivot in the active dart of minimum slack that
lies closest to a leaf of \(\cotree\);
in other words, we choose the minimum slack dart encountered first on the \(q\) to \(r\) path
through \(\cotree\).
Following the leafmost rule guarantees we always maintain \emph{leftmost} shortest path trees.

\subsection{Modified Perturbation Scheme and the Leafmost Rule}
\label{subsec:leafmost_modify}

We briefly return to the problem of computing multiple-source shortest paths around a face \(r\) in
a graph \(G = (V, E, F)\) of arbitrary genus~\(g\).
Consider the following slight modifications to our perturbation scheme described in
Section~\ref{sec:scheme}.
First, the drainage \(\coflow\) used to define the perturbed costs is required to use \(r\) as its
sink.
We can easily compute such a drainage using a dual spanning tree as described in
Section~\ref{sec:scheme}.
Second, we forgo the \(+1\) added as the second component to each dart's cost, instead only using
\(2g + 1\) integers based on the homology signatures and \(\coflow\) to perturb each dart's cost.
Let \(\cost' : \dartsof{E} \to \R \times \N^{2g + 1}\) be the resulting perturbed costs.

The first modification described above has no effect on our scheme's guarantee for minimum-cost
flows and shortest paths.
The second modification may have consequences, though.
Namely, our modified scheme may introduce negative cost directed cycles where the original cost
function had zero cost cycles, and minimum-cost flows and ``shortest paths'' may not be unique if a
dart and its reversal both have zero cost.
From this point forward, we will work under the assumption that every directed cycle has strictly
positive cost according to the unperturbed cost function \(\cost\).%
\footnote{Our assumption does not appear necessary for guaranteeing correctness or efficiency of
Eisenstat and Klein's~\cite{ek-lafms-13} planar graph multiple-source shortest paths algorithm.
Unlike the generalization presented in Section~\ref{sec:linear-time}, their algorithm may begin with
an \emph{arbitrary} shortest path tree.
In this case, they do not maintain a holiest shortest path tree, but their leafmost pivots do
provide the weaker but sufficient guarantee that shortest paths to a common endpoint do not cross.}

Now, let us return to the planar setting for the rest of this section (so \(g = 0\)).
We claim the pivots chosen using the leafmost rule with dart costs \(\cost\) are actually the same
as the unique pivot choices guaranteed using perturbed costs \(\cost'\).
We state this claim formally in the following lemma.
\begin{lemma}
  \label{lem:leafmost_explanation}
  Consider an iteration of the multiple-source shortest paths algorithm in planar graphs where the
  source of \(T\) moves from vertex \(u\) to vertex \(v\).
  For each value of \(\parameter\), the leafmost active dart of minimum slack according to \(\cost\)
  is also the unique active dart of minimum slack according to \(\cost'\).
\end{lemma}
\begin{proof}
  We can assume the lemma holds inductively across earlier iterations of the algorithm and earlier
  values of \(\parameter\) within the current iteration, meaning the current holiest path tree \(T\)
  is the same for the algorithm using the leafmost rule with \(\cost\) and the algorithm using
  perturbed costs \(\cost'\).
  Recall, the drainage \(\coflow\) used to define \(\cost'\) has sink \(r\).
  Lemma~\ref{lem:drainages} essentially states that the exact choice of \(\coflow\) does not
  matter;
  all drainages with equivalent dual imbalances result in the same unique minimum-cost flows and
  shortest paths in \(G\).
  Therefore, we can assume without loss of generality that \(\coflow\) is non-zero only within the
  complementary dual spanning tree \(\cotree\) of \(T\).

  Now, suppose there are two active darts \(x \arcto y\) and \(x' \arcto y'\) of minimum slack
  according to \(\cost\).
  No darts along \(\shortestpath(v, x)\), \(\shortestpath(v, x')\), \(\shortestpath(u, y)\), or
  \(\shortestpath(u, y')\) have a non-zero face part to their perturbed costs.
  Therefore, \(\coflow(\slack_\parameter(x \arcto y)) = \coflow(\cost'(x \arcto y) - \parameter)\),
  and \(\coflow(\slack_\parameter(x' \arcto y')) = \coflow(\cost'(x' \arcto y') - \parameter)\).
  That perturbation term is lower for whichever dart of \(x \arcto y\) and \(x' \arcto y'\) lies
  closer to the leaf of \(\cotree\).
\end{proof}

The intuition provided by Lemma~\ref{lem:leafmost_explanation} will prove crucial in the next
section where we describe our linear-time multiple-source shortest paths algorithm for small integer
costs in embedded graphs of constant genus.
While there does not appear to be an intuitive definition of leafmost in higher genus embedded
graphs, we \emph{do} have an equivalent perturbation scheme already defined for that setting.
Our goal will be to efficiently find the unique pivots guaranteed by our use of the scheme.

\paragraph*{Remark}
We briefly remark that our perturbation scheme also explains the use of the leafmost rule in
efficient \(s,t\)-maximum flow algorithms based on parametric shortest paths in the dual
graph~\cite{bk-amfdp-09,e-mfpsp-10,ek-lafms-13}.
For this problem, the drainage is actually computed in the \emph{primal} graph using \(t\) as the
sink.

\end{longversion}
\section{Linear-time Multiple-Source Shortest Paths for Small Integer Costs}
\label{sec:linear-time}

Let~\(G = (V,E,F)\) be an embedded graph of size~\(n\) and constant genus~\(g\) and let \(r \in F\) be a face
of \(G\).
Let \(\cost : \dartsof{E} \to \R\) be a dart cost function where each \(\cost(d)\) is a \emph{small
non-negative integer}.
Let \(L\) be the sum of the dart costs.
We now describe an algorithm for computing multiple-source shortest paths in this setting that runs
in \(O(g(g n \log g + L))\) time.
Like Eisenstat and Klein~\cite{ek-lafms-13}, we primarily focus on computing an initial shortest
path tree and then performing the pivots needed to move the source of the tree around~\(r\).
We will address computing shortest path distances later.

\begin{longversion}
Let \(\cost' : \dartsof{E} \to \R \times \N^{2g + 1}\) be the perturbed costs computed using the
\emph{modified} scheme presented in Section~\ref{subsec:leafmost_modify}.
Namely, \(\cost'\) is computed using a drainage \(\coflow\) with sink \(r\), and its \(2g + 1\)
perturbation terms are based only on homology signatures and \(\coflow\).
\end{longversion}
\begin{shortversion}
Our algorithm implicitly maintains holiest shortest path trees according to a slight modification of our
lexicographic perturbation scheme.
First, the drainage \(\coflow\) used to define the perturbed costs is required to use \(r\) as its
sink.
Second, we forgo the \(+1\) added as the second component to each dart's cost, instead only using
\(2g + 1\) integers based on the homology signatures and \(\coflow\) to perturb each dart's cost.
Let \(\cost' : \dartsof{E} \to \R \times \N^{2g + 1}\) be the resulting perturbed costs.
In order to guarantee shortest paths are still well defined and unique according to \(\cost'\), we
must make the small assumption that there are no cycles with zero unperturbed cost.%
\footnote{Our assumption does not appear necessary for guaranteeing correctness or efficiency of
Eisenstat and Klein's~\cite{ek-lafms-13} planar graph multiple-source shortest paths algorithm.
Unlike our generalization, their algorithm may begin with an \emph{arbitrary} shortest path tree.
In this case, they do not maintain a holiest shortest path tree, but their leafmost pivots do
provide the weaker but sufficient guarantee that shortest paths to a common endpoint do not cross.}
\end{shortversion}
Let \(T\) be the shortest path tree maintained while running Section~\ref{sec:MSSP}'s
multiple-source shortest paths algorithm along face \(r \in F\).
Suppose we are moving the source of \(T\) from vertex \(u\) to vertex \(v\), and let \(v \arcto u =
r \fencewith q\).
Let \(\cost'_\parameter\) be the cost function parameterized by \(\parameter\) as described in
Section~\ref{sec:MSSP}.

For planar graphs, Eisenstat and Klein~\cite{ek-lafms-13} explicitly maintain the slacks of the
darts based on the original cost function \(\cost\).
In other words, they maintain the first component of the slacks according to \(\cost'_\lambda\).
We will refer to these values as the \EMPH{original slacks} and slacks defined using every component
of \(\cost'_\lambda\) as the \EMPH{perturbed slacks}.
Recall, the edges outside of \(T\) form a spanning tree \(\cotree\) in \(G^*\) if \(G\) is planar.
To find pivots, Eisenstat and Klein walk a pointer up the directed dual path from \(q\) to \(r\) in
\(\cotree\).
When they find a dart \(x \arcto y\) with \(0\) original slack, they perform a pivot by adding \(x
\arcto y\) to \(T\) and removing the old predecessor dart \(w \arcto y\) from \(T\).
\begin{longversion}
According to Lemma~\ref{lem:leafmost_explanation}, this pivot is exactly the pivot \emph{required}
by using the perturbed slacks.
\end{longversion}
After performing the pivot, Eisenstat and Klein reset their pointer to continue the walk from the
first dart that only appears in the new \(q\) to \(r\) path in \(\cotree\).
If their walk reaches \(r\), then every dart along the current \(q\) to \(r\) path has positive
original slack.
They decrement the unperturbed slack values for every dart in the path, increment the unperturbed
slacks for those darts' reversals, and start a new walk from \(q\).

As described above, their algorithm does not appear to generalize cleanly to higher genus surfaces.
The dual complement to \(T\) is no longer a spanning tree, so it is not clear what route a pointer
should take.
In particular, it is completely unclear what the leafmost dart of \(0\) original slack should be,
especially when the set of active darts may not even be a connected subgraph of \(G^*\).
While leafmost may not cleanly generalize, however, our perturbation scheme is already defined for
higher genus embeddings.

\subsection{Preliminary Observations}
\label{subsec:linear-time_observations}
We now present some useful observations, slightly modifying conventions and terminology from
Erickson and Har-Peled~\cite{eh-ocsd-04} and Cabello \etal~\cite{cce-msspe-13}.
Let \(\cutgraph\) be the set of edges complementary to \(T\).
We refer to \(\cutgraph\) as a \EMPH{cut graph};
removing the dual embedding of \(X\) cuts the underlying surface \(\Sigma\) into a disk.
%One disk contains the blue vertices, and the other contains the red vertices.
Let \(\core\) be the \EMPH{2-core} of \(\cutgraph\) obtained by repeatedly removing vertices of
degree \(1\) except for \(q\) and \(r\) until no others remain.
%\begin{longversion}
We refer to the dual forest of removed edges as the \EMPH{hair} \(\hair\) of \(\cutgraph\).
%\end{longversion}
The 2-core \(\core\) consists of up to \(6g + 1\) dual paths \(\cutpath_1, \cutpath_2, \dots\) that
meet at up to \(4g + 2\) dual vertices (see Erickson and Har-Peled~\cite[Lemma 4.2]{eh-ocsd-04}).
Each of these \(4g + 2\) dual vertices except possibly \(q\) and \(r\) has degree at least \(3\).
We refer to the (unoriented) dual paths \(\cutpath_1, \cutpath_2, \dots\) as \EMPH{cut paths}.
%\begin{longversion}
%The endpoints of the cut paths are called \EMPH{branching vertices}.
%\end{longversion}
We let \(\cutpath_q\) denote the (possibly trivial) maximal subpath of \(\core\) with one
endpoint on \(q\) that contains at most one vertex that is either degree \(3\) or equal to \(r\).
We let~\(\cutdart_q\) denote the orientation of~\(\cutpath_q\) that begins with~\(q\).

The 2-core \(\core\) of \(\cutgraph\) is useful, because there is a subset of oriented cut paths
containing precisely the set of active darts~\cite[Section 4.2]{cce-msspe-13}.
In particular, let~\(\corewalk\) denote the \EMPH{blue boundary walk}, the clockwise facial walk
along dual darts of \(\core \cup \set{rq}\) that includes every primal dart with a blue tail in
\(G\) \emph{except for} \(r \fencewith q = v \arcto u\).
A dart~\(d\) is active if and only if \(d\) is in~\(\corewalk\) but \(\rev(d)\) \emph{is not}.
Either every dart in a given oriented cut path has this property, or none of them do.
See Figure~\ref{fig:core}.

\begin{figure}[t]
\centering
\begin{longversion}
\includegraphics[scale=0.4]{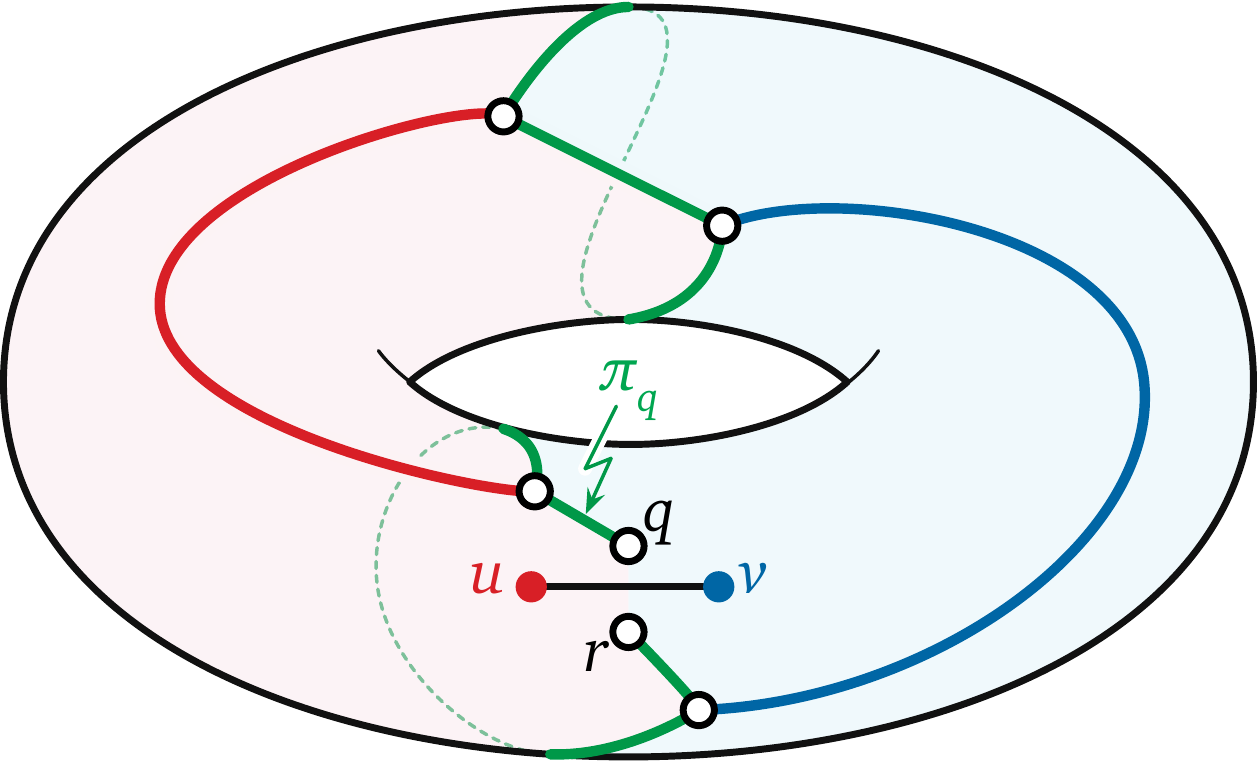}
\end{longversion}
\begin{shortversion}
\includegraphics[width=0.8\columnwidth]{Fig/torus}
\end{shortversion}
\caption{The 2-core \(\core\) of a cut graph \(\cutgraph\) on the torus.
Cut paths between red vertices are red, cut paths between blue vertices are blue, and cut paths containing active darts are green.}
\label{fig:core}
\end{figure}

%\begin{longversion}
An edge \(e \notin T\) is in \(\core\) if and only if it forms a non-contractible fundamental cycle
with \(T\) with respect to \(\Sigma - r\) or it forms a contractible fundamental cycle with \(T\)
and lies on \(\cutpath_q\) (see Cabello \etal~\cite[Section 4.1]{cce-msspe-13}).
We have the following lemmas, the first three of which follow immediately from the symmetry present
in the definitions of slack and our perturbation scheme.
\begin{shortversion}
We again refer the reader to the full version of this paper for proofs of the remaining lemmas.
\end{shortversion}
%\end{longversion}

\begin{lemma}
  \label{lem:zero_slack}
  Let \(d\) be any dart of~\(T\).
  We have
\begin{longversion}
  \[\homsig{\slack_\parameter(\rev(d))} \tuplecat \coflow(\slack_\parameter(\rev(d))) =
  \homsig{\slack_\parameter(d)} \tuplecat \coflow(\slack_\parameter(d)) = 0.\]
\end{longversion}
\begin{shortversion}
  \(\homsig{\slack_\parameter(\rev(d))}\)\linebreak
  \(\tuplecat \coflow(\slack_\parameter(\rev(d))) = \homsig{\slack_\parameter(d)} \tuplecat
  \coflow(\slack_\parameter(d)) = 0.\)
\end{shortversion}
\end{lemma}

\begin{lemma}
  \label{lem:slack_homology}
  Let \(d\) be an active dart.
\begin{longversion}
  We have \(\homsig{\slack_\parameter(d)} = \homsig{\fundcycle(T, d)} + \homsig{\cost'(v \arcto u)}
  - \homsig{\parameter}\).
\end{longversion}
\begin{shortversion}
  We have \(\homsig{\slack_\parameter(d)}\)\linebreak
  \(= \homsig{\fundcycle(T, d)} + \homsig{\cost'(v \arcto u)}
  - \homsig{\parameter}\).
\end{shortversion}
\end{lemma}
%begin{proof}
  %Let \(\flow = (\shortestpath(v, x) \concat x \arcto y) - (u \arcto y \concat \shortestpath(u, y))\).
  %Signature \(h\) is equal to the homology part of \(\cost'(\flow)\).
  %However, \(\cost'(\flow) = \slack_\lambda(d) + \cost'(v \arcto u) - \lambda\).
%\end{proof}

%\begin{lemma}
  %\label{lem:greater_slack_homology}
  %Let \(d = x \arcto y\) be an active dart with \(\slack_\parameter(d)_0 = 0\).
  %We have \(\homsig{\fundcycle(T, d)} \geq \homsig{\cost'(v \arcto u) - \parameter}\).
%\end{lemma}

\begin{lemma}
  \label{lem:seek_homology}
  Let \(D = \seq{d_1, d_2, \dots}\) be a sequence of the active darts in increasing lexicographic
  order of \(\homsig{\slack_\parameter(d_i)}\).
  Sequence~\(D\) is in increasing lexicographic order of \(\homsig{\fundcycle(T, d_i)}\) as well.
\end{lemma}

\begin{lemma}
  \label{lem:change_homology}
  Suppose a pivot inserts dart~\(d^+\) into the holiest path tree~\(T\) while removing dart~\(d^-\).
  Let \(d = x \arcto y\) be any dart where~\(x\) turns from red to blue while~\(y\) does not change
  colors.
  We have~\(\homsig{\fundcycle(T - d^- + d^+, d)} = \homsig{\fundcycle(T, d)} +
  \homsig{\fundcycle(T, d^+)}\).
\end{lemma}
\begin{longversion}
\begin{proof}
  We have~\(\homsig{\fundcycle(T - d^- + d^+, d)} = \homsig{\fundcycle(T, d^+)} +
  \homsig{\fundcycle(T, d)} - \homsig{\fundcycle(T, d^-)}\).
  However,~\(\homsig{\fundcycle(T, d^-)} = 0\).
\end{proof}
\end{longversion}

\begin{lemma}
  \label{lem:fundamental_cycle_sigs}
  Fix an oriented cut path~\(\cutdart\), and let \(d_1\) and \(d_2\) be darts of~\(\cutdart\).
  Cycles \(\cycle_1 = \fundcycle(T, d_1)\) and \(\cycle_2 = \fundcycle(T, d_2)\) are homologous.
\end{lemma}
\begin{longversion}
\begin{proof}
  Let \((T, \leftovers, \cotree)\) be a tree-cotree decomposition of \(G\) where \(\cutgraph\)
  contains all edges of \(\cotree\) and \(\leftovers\).
  Without loss of generality, assume the homology signatures \(\homsig{\cdot}\) are defined using
  \((T, \leftovers, \cotree)\).
  For any set of dual cycles \(\gamma^*_1, \gamma^*_2, \dots\) in \(\cutgraph\) and edge \(e\) in
  \(\cutgraph\), the set of cycles passing through \(e\) is determined by which cut path \(e\)
  belongs to.
  Therefore \([d_1] = [d_2]\).
  Because \([d] = 0\) for every dart in \(T\) and its reversal, we have \([\cycle_1] =
  [\cycle_2]\).
  Lemma~\ref{lem:signatures} implies \(\cycle_1\) and \(\cycle_2\) are homologous.
\end{proof}
\end{longversion}

\begin{lemma}
  \label{lem:cut_path_order}
  Let \(\cutdarts = \set{\cutdart_{i_1}, \cutdart_{i_2}, \dots}\) be the set of oriented cut paths
  containing darts that are \(\emph{active}\) and whose perturbed slacks have equal homology part.
  Finally, let \(D = \seq{d_1, d_2, \dots}\) be the sequence of darts within these oriented cut
  paths in increasing order of perturbed slack.
  We have \(D\) is a subsequence of the blue boundary walk~\(\corewalk\).
  In particular, the darts within any one oriented cut path \(\cutpath \in \cutpaths\) appear as a
  consecutive subsequence of \(D\).
\end{lemma}
\begin{longversion}
\begin{proof}
  We assume \(D\) contains at least two darts; otherwise, the proof is trivial.
  As in the proof of Lemma~\ref{lem:leafmost_explanation}, we assume without loss of generality that
  the drainage~\(\coflow\) used to define~\(\cost'\) is defined using cotree~\(\cotree\) of some
  tree-cotree decomposition \((T, \leftovers, \cotree)\) where \(\cutgraph = \leftovers \cup
  \cotree\).
  In particular, for each active dart \(d\), we have \(\coflow(\slack_\parameter(d)) =
  \coflow(\cost'(d)) - \coflow(\parameter)\).

  Let~\(d_{i_1} = x_1 \arcto y_1\) and~\(d_{i_2} = x_2 \arcto y_2\) be two distinct active darts
  in~\(D\) where~\(\corewalk\) includes \(d_{i_1}\) before \(d_{i_2}\).
  Let \(\walk_1 = \shortestpath(v, x_1) \concat x_1 \arcto y_1 \concat \rev(\shortestpath(u, y_1))\)
  and \(\walk_2 = \shortestpath(v, x_2) \concat x_2 \arcto y_2 \concat \rev(\shortestpath(u,
  y_2))\).
  Let \(P_1'\) be the maximal consecutive subsequence of \(P_1\) that does not contain darts of
  \(P_2\) and define \(P_2'\) similarly.
  Finally, let \(\cycle = P_2' \concat \rev(P_1')\).

  By assumption, \(\homsig{\fundcycle(T, d_{i_1})} = \homsig{\fundcycle(T, d_{i_2})}\).
  Therefore, \(\homsig{\cycle} = 0\).
  Also, \(\cycle\) is simple and non-empty.
  Together with Lemma~\ref{lem:signatures} these facts imply~\(\cycle\)'s dual darts enter a strict
  subset of faces~\(F' \subseteq F\).
  Blue boundary walk~\(\corewalk\) crosses~\(\cycle\) only at~\(\rev(d_{i_1})\) and~\(d_{i_2}\)
  before ending at~\(r\), implying~\(r \in F'\).
  Lemmas~\ref{lem:cost_components} and~\ref{lem:signatures} imply \(
  \coflow(\slack_\parameter(d_{i_2})) - \coflow(\slack_\parameter(d_{i_1})) =
  \coflow(\cost'(d_{i_2})) - \coflow(\cost'(d_{i_1})) = \coflow(\cost'(\cycle)) > 0\).
\end{proof}
\end{longversion}

\begin{lemma}
  \label{lem:final_parameter}
  Suppose \(\parameter_0 \tuplecat \homsig{\parameter} = \cost'(v \arcto u)_0 \tuplecat
  \homsig{\cost'(v \arcto u)}\) and \(\coflow(\parameter) \leq \coflow(\cost'(v \arcto u))\).
  Then, there are no more pivots to perform in the current iteration.
  In particular, for any dart~\(d\) with~\(\slack_\parameter(d)_0 = 0\), we
  have~\(\coflow(\slack_\parameter(d)) > \coflow(\cost'(v \arcto u)) - \coflow(\parameter)\).
\end{lemma}
\begin{longversion}
\begin{proof}
  Consider any active dart \(d\).
  The homology and face parts of \(\slack_\parameter(d)\) are equal to the their counterparts in the
  cost of \(\gamma = \shortestpath(v, x) \concat x \arcto y \concat \rev(\shortestpath(u, y))
  \concat u \arcto v\) according to~\(\cost'_\parameter\).
  If \(\homsig{\slack_\parameter(d)}\) is positive, then \(d\) will never be tense in
  the current iteration.
  Therefore, \(\homsig{\cost'_\parameter(\gamma)}\) must be zero for \(d\) to become tense in the
  current iteration.
  However, \(\homsig{\cost'_\parameter(\gamma)} = \homsig{\cost'(\gamma)}\).
  By Lemmas~\ref{lem:cost_components} and~\ref{lem:signatures}, \(\gamma\) is a boundary
  circulation.
  Because \(q \fencewith r\) is in \(\gamma\), the dual darts of \(\gamma\) must enter a strict
  subset of faces~\(F' \subseteq F\) containing~\(r\).
  Lemma~\ref{lem:drainages} implies~\(\coflow(\cost'(\gamma)) > 0\).
  In particular, \(\parameter\) cannot rise high enough to make \(d\) tense before it
  reaches~\(\cost'(v \arcto u)\).
\end{proof}
\end{longversion}

\subsection{Algorithm Outline}
\label{subsec:linear-time_outline}
Based on the previous observations, we use the following strategy to compute multiple-source
shortest paths.
Recall, in each \emph{iteration} of the multiple-source shortest paths algorithm, we move the source
of \(T\) between consecutive vertices \(u\) and \(v\) on \(r\).
Each iteration begins with the special pivot which sets \(\parameter := -\dist_{\cost'}(u, v)\).
Now, consider continuously increasing \(\parameter\) until it reaches \(\cost'(v \arcto u)\).

We divide the remainder of the iteration into a number of \EMPH{rounds}.
Over the course of each round, \(\parameter_0\) remains a static integer while the homology and face
parts of \(\parameter\) continuously increase.
The round ends when either \(\parameter = \cost'(v \arcto u)\) or the homology and face parts of
\(\parameter\) become infinitely positive.
If the latter case occurs, we say the round is \EMPH{fully completed}.
At the end of fully completed rounds, \(\parameter_0\) increases by \(1\), and the homology and face
parts of \(\parameter\) become infinitely negative.
We perform all pivots necessary to increase \(\parameter\) to the end of each round in the order
these pivots occur.
%, but it only explicitly records the whole integer changes to the
%unperturbed parts of \(\parameter\) and each darts' slack.

An active dart~\(d\) can pivot into~\(T\) during the current round only if \(\slack_\parameter(d)_0
= 0\) which implies \(\homsig{\slack_\parameter(d)} \tuplecat \coflow(\slack_\parameter(d)) \geq
0\);
the pivot occurs when \(\parameter\) becomes large enough that the entire slack vector goes to
\(0\).
Therefore, we check each active dart \(d\) at the moment that \(\homsig{\slack_\parameter(d)}
\tuplecat \coflow(\slack_\parameter(d)) = 0\) to see if \(\slack_\parameter(d)_0 = 0\) as well.
Between pivots, and in accordance with Lemmas~\ref{lem:slack_homology}--\ref{lem:cut_path_order}, we
do these checks cut path-by-cut path, first in lexicographic order by \(\homsig{\fundcycle(T, d)}\),
and then in the order they appear along the blue boundary walk~\(\corewalk\).
We check the active darts within each of the oriented cut paths in order.
When we detect a pivot must occur, we perform the pivot and then resume checking darts that still
have non-negative homology part to their slack.
By Lemma~\ref{lem:zero_slack}, the first of these checks occurs on the reversal of the dart just
removed from~\(T\) in the previous pivot.
We say a dart \(d\) has been \EMPH{passed} in the current round when \(\homsig{\slack_\parameter(d)}
\tuplecat \coflow(\slack_\parameter(d)) < 0\).
Each round of checking and pivoting is further broken into three \EMPH{stages} as follows:
\begin{enumerate}
  \item
    During the first stage, \(\homsig{\parameter} < \homsig{\cost'(v \arcto u)}\).
    Following Lemma~\ref{lem:slack_homology}, we only check for pivots along cut paths whose
    darts~\(d\) have \(\homsig{\fundcycle(T, d)} < 0\).
    By Lemma~\ref{lem:final_parameter}, the current iteration finishes at the end of this stage if
    \(\parameter_0 = \cost'(v \arcto u)_0\).
  \item
    At the beginning of this stage, the homology and face parts of \(\parameter\) are equal to the
    homology and face parts of \(\cost'(v \arcto u)\).
    To increase \(\homsig{\parameter}\), we must check for pivots from active darts~\(d\) where
    \(\homsig{\fundcycle(T, d)} = 0\).
    By Lemma~\ref{lem:cut_path_order}, darts in~\(\cutdart_q\) must be checked first.
    The stage ends \emph{immediately} after we verify~\(\slack_\parameter(d)_0 > 0\) for every
    dart~\(d\) in~\(\cutdart_q\).
    This stage requires a bit of care, because we cannot afford to explicitly update~\(\cutpath_q\)
    after every pivot.
    This stage most closely resembles how Eisenstat and Klein~\cite{ek-lafms-13} handle each of the
    rounds as defined above, because it is the only non-trivial stage when~\(G\) is a planar graph.
  \item
    At the beginning of this stage, \(\homsig{\parameter} = \homsig{\cost'(v \arcto u)}\).
    We now check for pivots along cut paths other than~\(\cutpath_q\) whose darts~\(d\) have
    \(\homsig{\fundcycle(T, d)} \geq 0\).
    The end of this stage marks a full completion of the current round.
\end{enumerate}

The rest of this section is organized as follows.
In Section~\ref{subsec:linear-time_structures}, we go over the data structures used to efficiently
implement our algorithm.
We then discuss checking for and performing pivots during stages 1 and 3
(Section~\ref{subsec:linear-time_stages1and3}) and stage 2 (Section~\ref{subsec:linear-time_stage2})
separately before discussing the special pivot (Section~\ref{subsec:linear-time_special-pivot}) in
more detail.
We discuss how to efficiently pick which cut path to perform pivot checks on in
Section~\ref{subsec:linear-time_searching-cut-darts} before finally analyzing the running time of
our algorithm in Section~\ref{subsec:linear-time_running-time}.

\subsection{Data Structures}
\label{subsec:linear-time_structures}

We begin by describing the data structures used by our algorithm.
These data structures extend the ones used by Eisenstat and Klein~\cite{ek-lafms-13} to work with
more general embedded graphs.
As a general rule, we explicitly store lots of information about cut paths other than
\(\cutpath_q\).
Handling \(\cutpath_q\) requires more care as we do not have time to explicitly maintain it or even
remember both of its endpoints at certain moments in the algorithm.
\begin{itemize}
  \item
    For each vertex \(v \in V\), we store its predecessor dart \(pred(v)\) in the shortest path tree
    \(T\).
  \item
    For each dual vertex \(p \in F\) we store its successor dart \(succ(p)\) in an arbitrary dual
    spanning tree of \(\cutgraph\) rooted at \(r\) as well as a boolean \(visited(p)\).
    These values will aid us in maintaining an implicit list of darts along \(\cutdart_q\).
    From the end of stage 2 to immediately before the end of stage 1, \(visited(p)\) is
    \(\textsc{True}\) if \(succ(p)\) lies on \(\cutdart_q\).
    At the end of stage 1 and between iterations, \(visited(p)\) is set to \(\textsc{False}\) for
    every dual vertex~\(p\).
    During stage 2, it is only \(\textsc{True}\) if \(succ(p)\) lies on \(\cutdart_q\) \emph{and} it
    has been passed in the current round.
  \item
    We maintain the unperturbed part of parameter~\(\parameter\) as \(\uparameter\).
    %\note{Kyle: This seems backwards. See the last notes. :(}
  \item
    We maintain a \EMPH{reduced cut graph}~\(\reducedcutgraph\), an embedded graph of genus~\(g\)
    with a vertex for every vertex of the 2-core \(\core\) and an edge for every cut
    path.
    We refer to vertices and edges of \(\reducedcutgraph\) as cut vertices and cut edges,
    respectively.
    As in all embedded graphs, each cut edge~\(\cutedge_i\) has two \EMPH{cut darts} representing
    the two orientations of \(\cutpath_i\).
    We say a cut dart is~\EMPH{active} if its oriented cut path contains active darts.
  \item
    For each dart \(d \in \dartsof{E}\) we store its \emph{unperturbed} slack \(\uslack(d) :=
    \slack_\parameter(d)_0\) as an integer.
    %\note{Kyle: Use \(\slack'\) everywhere we're doing perturbations?}
    If \(d\)'s edge lies on a cut path other than \(\cutpath_q\), we also store \(\cutdart(d)\), the
    cut dart for \(d\)'s oriented cut path;
    otherwise, we set \(\cutdart(d)\) to \(\textsc{Null}\).
  \item
    Lemma~\ref{lem:fundamental_cycle_sigs} implies that for each cut dart~\(\cutdart\), there is a
    unique \EMPH{fundamental cycle homology signature} equal to \(\homsig{d}\) for any dart~\(d\)
    in~\(\cutdart\)'s oriented cut path.
    We store this value as~\(\homsig{\cutdart}\).
    We also store a boolean \(passed(\cutdart)\) that is \(\textsc{True}\) if \(\cutdart\) is active
    and has been passed in the current round.
    Finally, if~\(\cutdart\)'s cut path is not \(\cutpath_q\), then we store its darts in a
    list \(darts(\cutdart)\).
%  \item
%    We maintain a balanced binary tree~\(\darttree\) containing ordered lists of cut darts.
%    The cut darts within each list~\(\dartlist\) have darts with the same fundamental cycle homology
%    siganture, and each cut dart~\(\cutdart\) has a pointer \(\dartlist(\cutdart)\) back to its
%    tree node.
%    The lists for negative fundamental cycle homology signatures \emph{always} come after lists for
%    non-negative signatures.
%    Otherwise, lists are ordered by the signatures lexicographically.
  \item
    Finally, we maintain a \EMPH{finger}~\(\finger\) which points to a single dual vertex.
    When searching \(\cutdart_q\) for the next dart to pivot into \(T\) we will use \(\finger\) to
    track our progress.
\end{itemize}

We begin the algorithm by computing the initial holiest tree \(T\) rooted at some vertex \(u\) on
\(r\).
We cannot simply apply the linear time algorithm of Henzinger \etal~\cite{hkrs-fspap-97}, because
the perturbed cost of some darts may be negative.
Therefore, we begin by computing \emph{some} shortest path tree rooted at \(r\) using the
unperturbed costs~\(\cost\) in~\(O(n)\) time.
Let~\(H \subseteq \dartsof{E}\) be the subset of darts with~\(0\) unperturbed slack.
The holiest tree uses only darts of~\(H\).
Further, our assumption that~\(G\) contains no zero-cost cycles guarantees~\(H\) is acyclic.
We compute the holiest tree~\(T\) from~\(u\) in~\(O(gn)\) time using the standard shortest path tree
algorithm for directed acyclic graphs on~\(H\).
The initial computation of~\(T\) is the \emph{only} time our algorithm will explicitly refer to the
face parts of the darts' perturbed costs.
The rest of the data structures can be initialized in~\(O(gn)\) time.
%There are no active darts, so the binary search tree~\(\darttree\) is initially empty.
%The finger \(\finger\) is initially set to \(\textsc{Null}\).
%All that remains is to describe how to find the next dart to pivot into~\(T\), and how to perform
%the pivot efficiently.
We now discuss how to handle each of the stages as described above.
Performing the special pivot is handled last, because the procedure closely resembles pivots
performed during these stages.

\subsection{Stages 1 and 3}
\label{subsec:linear-time_stages1and3}

In both stages 1 and 3, we check for and perform pivots using darts on cut paths other
than~\(\cutpath_q\).
Suppose we have just performed a pivot or the current stage has just begun.
We begin by discussing how to check for pivots.

\paragraph*{Checking for pivots}
We find the oriented cut path~\(\cutdart\) containing the next dart that needs checking for a pivot,
assuming one exists.
The search can be accomplished by taking a walk along the cut darts of the reduced cut graph that
correspond to the blue boundary walk~\(\corewalk\).
The active cut darts are those whose edges are encountered once during this walk.
The oriented cut path we seek corresponds to the first active cut dart~\(\cutdart\) in the walk with
lexicographically least fundamental cycle homology signature among all unpassed active cut darts.
Later, we describe a more efficient way to find \(\cutdart\).

If~\(\cutdart = \cutdart_q\), then we have completed stage 1.
We take a walk in the dual graph from~\(q\), following \(succ\) pointers until we reach a dual
vertex \(p\) for which \(visited(p) \neq \textsc{True}\).
We unset the \(visited\) booleans for each dual vertex encountered during the walk.
The \(visited\) booleans will be fixed to accurately represent~\(\cutdart_q\) before the next run of
stage 1 or 3.
If there is no choice for~\(\cutdart\), because every active dart has been passed, then we have
completed stage 3.
We increment~\(\uparameter\) and unset the \(passed\) variables for every cut dart.
Then, for each active dart \(d\), we decrement \(\uslack(d)\) and increment \(\uslack(\rev(d))\).
%The slack changes can be implemented in~\(O(k + g)\) time where~\(k\) is the number of active darts
%by doing another walk in the reduced cut graph to find precisely which cut paths have active darts.

If the current stage is still active, then we check the darts of \(darts(\cutdart)\) in
order, performing a pivot if we find a dart \(d\) in \(darts(\cutdart)\) with \(\uslack(d) = 0\).
Suppose dart~\(d^-\) was removed from~\(T\) in the previous pivot, we have not yet completed a round
since that pivot, and~\(\rev(d^-)\) is in \(darts(\cutdart)\).
In this case,~\(\rev(d^-)\) is the first dart we check.%
\footnote{As discussed above, we start at~\(\rev(d^-)\), because the homology and face parts of its
slack are equal to~\(0\).
Our analysis still goes through if we check all the earlier darts of~\(darts(\cutdart)\)
as well, although we will not find any pivots that use those earlier darts.}
Otherwise, we start with the first member of \(darts(\cutdart)\).
If we find no dart~\(d\) that can be pivoted into~\(T\), we set \(passed(\cutdart)\) to
\(\textsc{True}\) and repeat the search procedure for the next \(\cutdart\).

\paragraph*{Performing a pivot}
Suppose we decide to pivot dart \(d^+ = x^+ \arcto y = o^+ \fencewith p^+\) into the holiest tree
\(T\) while removing dart \(d^- = x^- \arcto y = o^- \fencewith p^-\).
Let~\(\cutdart^+\) be the cut dart whose oriented cut path contains~\(d^+\).
Dart~\(d^+\) lies on some cut path \(\cutpath \neq \cutpath_q\).
Cycle \(\fundcycle(T - d^- + d^+, d^-) = \fundcycle(T, d^+)\) is non-contractible in the surface
\(\Sigma - r\).
Therefore, \(d^-\) will belong to a new cut path other than \(\cutpath_q\) after the pivot.
We walk back through the hair of \(\cutgraph\) from \(p^-\) following \(succ\) darts
until we encounter a dual vertex~\(p_1\) such that either \(\cutdart(succ(p_1))\) is set or
\(visited(p_1)\) is set.
In the first case,~\(p_1\) lies on a cut path other than~\(\cutpath_q\).
In the latter case, either~\(p_1 = q\) or it lies interior to~\(\cutpath_q\).
We also walk forward through the hair from \(o^-\) following \(succ\) darts until we encounter a
dual vertex~\(p_2\) lying on some cut path.
Note we cannot have both walks end at a \(visited\) dual vertex.
Otherwise, we will create a contractible dual cycle after adding \(d^-\)'s edge to \(\cutgraph\),
implying \(T + d^- + d^+\) is not connected.

After finding \(p_1\) and \(p_2\), we modify the reduced cut graph \(\reducedcutgraph\) according to
the new cut path from~\(p_1\) to~\(p_2\) that we found.
Let \(\cutedge^-\) be the new cut edge for this cut path.
We add \(\cutedge^-\) to \(\reducedcutgraph\), possibly subdividing existing cut edges according to
where we found \(p_1\) and \(p_2\).
When we subdivide cut edges other than the one for~\(\cutpath_q\), we split the old cut edge's cut
dart's dart lists.
If we subdivide the cut edge for~\(\cutpath_q\), we unset \(visited\) for each dual vertex no longer
on~\(\cutpath_q\) and build new dart lists from scratch for the two new cut darts that are not
orientations of the newly shortened~\(\cutpath_q\).
The fundamental cycle homology signature for these new cut darts are initially set to~\(0\) to
continue respecting the fundamental cycles of~\(T\).
We create new dart lists for both cut darts of \(\cutedge^-\) by just including every dart we
encountered during the walks and their reversals.

We must then \emph{remove} the cut edge~\(\cutedge^+\) for~\(\cutdart^+\).
In a reversal of the above steps, we do a walk from~\(o^+\) along darts
of~\(darts(\rev(\cutdart^+))\) until we encounter a dual cut vertex~\(p_3\) of~\(\core\).
Every dual vertex \(p \neq p_3\) encountered during the walk becomes a hair of~\(\cutgraph\), so we
set \(succ(p)\) for each of these dual vertices to follow the walk.
If~\(p_3\) is an endpoint of~\(\cutpath_q\), \(\cutpath^+\), and exactly one other cut
path~\(\cutpath\), then the walk continues along~\(\cutpath\) until another vertex of~\(\core\) is
encountered, except now each dual vertex~\(p\) has both \(succ(p)\) set to follow the walk and
\(visited(p)\) set to~\(\textsc{True}\) to represent~\(\cutpath_q\) being enlongated to the new
endpoint.
A similar walk and setting of dual vertex variables is performed from~\(p^+\) using darts
of~\(darts(\cutdart^+)\).
We then update~\(\reducedcutgraph\) by removing~\(\cutedge^+\), changing the endpoint
of~\(\cutpath_q\) if necessary, and merging any cut edges sharing degree~\(2\) vertices
of~\(\reducedcutgraph\) as well as their lists.
Finally, we add~\(d^+\) to~\(T\) by setting \(pred(y) := d^+\).

Call each vertex turning from red to blue during the current pivot \EMPH{\color{purple} purple}.
We must now compute new fundamental cycle homology signatures for cut darts whose darts
in~\(G\) have one purple endpoint.
Observe,~\(\homsig{\fundcycle(T, d^+)} = \homsig{\cutdart^+}\).
Following Lemma~\ref{lem:change_homology}, we reassign~\(\homsig{\cutdart} := \homsig{\cutdart} +
\homsig{\fundcycle(T, d^+)}\) for each cut dart~\(\cutdart\) whose oriented cut edge contains darts
with purple tails but not purple heads.
The reversals of these cut darts are assigned the opposite fundamental cycle homology signatures.

Finally, we need to figure out which cut darts have been passed in the current round.
Cut dart~\(\rev(d^-)\) has now become active, because its tail is purple but not its head.
Let~\(\cutdart^-\) be the orientation of~\(\cutedge^-\) whose oriented cut path contains~\(d^-\).
Similar to above, we walk along the~\(O(g)\) cut darts that correspond to the blue boundary
walk~\(\corewalk\).
By Lemma~\ref{lem:zero_slack}, \(\homsig{\slack_\parameter(\rev(d))} \tuplecat
\coflow(\slack_\parameter(\rev(d))) = 0\).
Following Lemmas~\ref{lem:seek_homology} and~\ref{lem:cut_path_order}, we set \(passed(\cutdart) :=
\textsc{True}\) for each cut dart \(\cutdart\) such that either \(\homsig{\cutdart} <
\homsig{\rev(\cutdart^-)}\) or \(\homsig{\cutdart} = \homsig{\rev(\cutdart^-)}\) and \(\cutdart\)
appears earlier in the walk than~\(\rev(\cutdart^-)\).

\subsection{Stage 2}
\label{subsec:linear-time_stage2}

In stage 2, we check for and performs pivots using darts on~\(\cutdart_q\).
We begin by setting the finger \(f := q\) and then start checking for pivots.

\paragraph*{Checking for pivots}

We do the following iterative procedure to discover the darts of~\(\cutdart_q\) and look for pivots.
If \(\cutdart(succ(f))\) is set, we have discovered the dual vertex \(f\) where \(\cutpath_q\) is
incident to other cut paths as well as which of the currently stored cut edges contains~\(f\).
We update the reduced cut graph by moving~\(\cutpath_q\)'s endpoint to this newly discovered
intersection, subdividing or merging cut edges and cut darts' dart lists as discussed above.
Afterward, the current stage is over.

Suppose \(\cutdart(succ(f))\) is not set.
Then, we set \(visited(f)\) to true, because \(succ(f)\) must lie on~\(\cutdart_q\).
We then check if \(\uslack(succ(f)) = 0\).
If so, we pivot \(succ(f)\) into~\(T\) as discussed below.
If not, we set \(f\) to the head of \(succ(f)\) to continue the search for the next pivot.

\paragraph*{Performing a pivot}
Suppose we decide to pivot dart \(d^+ = x^+ \arcto y = o^+ \fencewith p^+\) into the holiest tree
\(T\) while removing dart \(d^- = x^- \arcto y = o^- \fencewith p^-\).
We need to reassign \(succ\) pointers and unset some \(visited\) booleans based on the changing cut
path~\(\cutpath_q\).
Dart \(\rev(d^-)\) will lie on~\(\cutdart_q\), because otherwise there will be no path from~\(o^+\)
to \(r\) in the cut graph.
We add \(d^-\)'s edge to \(\cutgraph\) by setting \(succ(p^-) := \rev(d^-)\).
Then, we walk from~\(p^-\), following \(succ\) darts until we encounter~\(o^+\).
We flip the \(visited\) boolean for every dart head encountered on this walk and reverse each dart
we walk along to reflect the new route for~\(\cutdart_q\) through \(\rev(d^-)\).
Finally, we add~\(d^+\) to~\(T\) by setting \(pred(y) := d^+\).
We set~\(f := p^-\) so that~\(\rev(d^-)\) will be the next dart checked for a pivot.%
\footnote{As in the previous footnote, we can afford to check each new dart up to~\(\rev(d^-)\)
on~\(\cutdart_q\) for a pivot as well, but doing so is unnecessary.}

\subsection{The Special Pivot}
\label{subsec:linear-time_special-pivot}
At the beginning of the iteration, we pivot dart \(v \arcto u = r \fencewith q\) into the holiest
tree \(T\) while removing dart \(d^- = x^- \arcto v = o^- \fencewith p^-\).
We consider two cases for how to handle the special pivot depending on whether it more closely
resembles a normal pivot during stages 1 or 3 or a normal pivot during stage 2.
We first remark that~\(\cutpath_q\) may not be known before the special pivot occurs, because all
the \(visited\) booleans are unset and \(q\) was just reassigned for the current iteration.
However,~\(r \fencewith q\) lies on a cut path other than~\(\cutpath_q\) if and only if it did so
immediately before reassigning \(q\) and performing the special pivot.

Suppose~\(r \fencewith q\) lies on a cut path~\(\cutpath^+\) other than~\(\cutpath_q\).
In this case,~\(\cutpath_q\) is actually trivial, because~\(q\) lies on~\(\cutpath^+\) as well.
In particular, the \(visited\) booleans are accurately set to be \(\textsc{False}\) everywhere.
We update the reduced cut graph by subdividing~\(\cutpath^+\) and its dart lists to represent the
location of~\(q\).
We then follow the same strategy as when we perform a pivot in stage 1 or 3, except we
interpret~\(q\) as~\(o^+\) and~\(r\) as~\(p^+\) (with this interpretation, \(o^+ \fencewith p^+\)
has the same primal head as \(d^-\)).
Of course, we actually set \(pred(u) := v\) at the end of the process instead of the other way
around.
As before, the process ends with~\(d^-\) belonging to a cut path other than~\(\cutpath_q\).
By Lemma~\ref{lem:zero_slack}, we have \(\homsig{\slack_\parameter(\rev(d^-))} \tuplecat
\coflow(\slack_\parameter(\rev(d^-))) = 0\).
Together with Lemma~\ref{lem:final_parameter}, we see the algorithm is now in the middle of stage 1
or stage 3.

Now, suppose~\(r \fencewith q\) is not on a cut path other than~\(\cutpath_q\).
Similar to the strategy for stage 2, we set \(succ(p^-) := \rev(d^-)\) and then we follow \(succ\)
darts starting at~\(p^-\) until we reach~\(q\).
We reverse all the \(succ\) pointers along the way and set \(visited\) to true for each dual vertex
at the tail of the new \(succ\) pointers to represent the new route~\(\cutdart_q\) must take
through~\(\rev(d^-)\).
We set \(pred(u) := v\).
By Lemmas~\ref{lem:zero_slack} and~\ref{lem:final_parameter}, the algorithm must be in the middle of
stage 2, so we set \(f := p^-\).

\subsection{Searching for Cut Darts}
\label{subsec:linear-time_searching-cut-darts}

At this point, we have a functioning algorithm that, as argued below, runs in linear time assuming
\(g\) is a constant.
However, the algorithm spends~\(O(g^2)\) time every time it starts searching a new dart list of a
cut dart.
To improve the running time, we maintain an ordered dictionary~\(\darttree\) containing active cut
darts.
The cut darts within~\(\darttree\) are sorted first in lexicographic order by their fundamental
cycle homology signatures and then in the order their darts of~\(G^*\) appear in the blue boundary
walk~\(\corewalk\).
Dictionary~\(\darttree\) will contain at most~\(O(g)\) cut darts at any one time.
It should support insertion and deletion in~\(O(g \log g)\) time each, finding the first entry in
constant time, and finding the successor of \emph{any of its current entries} in constant time.

Observe that a prefix of the cut darts in~\(\darttree\) have been passed.
We maintain a pointer to the first member of~\(\darttree\) that has not been passed.
Therefore, we can select a cut dart along which to do pivot checks in constant time.
If no pivots are found, then the pointer is moved to the next member of the dictionary.
If there is no next member, then a round has been fully completed.
Slack changes are performed in time linear in the number of active darts by touching only darts and
their reversals from~\(\cutdart_q\) and cut darts in~\(\darttree\).
Afterward, the pointer is moved to the first member of the dictionary.
The dictionary is updated after a special pivot that puts the algorithm in stage 1 or 3, after
completing a stage 2 that starts immediately after the special pivot, after each pivot in stages 1
and 3, and after any other stage 2 that performs a pivot.

For updates immediately following the special pivot or a stage 2 immediately following it, we remove
every cut dart from~\(\darttree\), perform the walk in the reduced cut graph corresponding to the
blue boundary walk, and add each active cut dart we encounter.
The other updates occur when the reduced cut graph is changing.
Each cut dart leaving the reduced cut graph is removed from the dictionary~\(\darttree\).
Then, a walk is done in the new reduced cut graph adding each new cut dart encountered
to~\(\darttree\).
The pointer is updated to the cut dart containing the reverse of the dart of~\(G\) just pivoted out
of~\(T\).

\subsection{Running Time Analysis}
\label{subsec:linear-time_running-time}

Having described how our algorithm is implemented, we turn to bounding its running time.
From prior work, we already know there are at most \(O(gn)\) pivots~\cite{cce-msspe-13}.
However, we still need a bound on the time spent interacting with individual darts in \(\cutgraph\)
that do not get pivoted into \(T\) as well as the time spent working with the reduced cut graph and
cut dart dictionary~\(\darttree\).

Fix a dart \(x \arcto y = o \fencewith p\).
Let \(\seq{u_1, u_2, \dots, u_k}\) be the sequence of sources for the shortest path tree \(T\) in
order around \(r\).
We say two paths \(\walk_1\) and \(\walk_2\) from a vertex on \(r\) to \(x\) are \EMPH{restricted
homotopic} with respect to \(x \arcto y\) if there is a homotopy from \(\walk_1\) to \(\walk_2\) in
\(\Sigma_o = \Sigma - \set{o, r}\) where the first endpoint of the path may move forward or backward
along the path \(u_1 \arcto u_2 \arcto \dots \arcto u_{k-1}\).
\begin{longversion}
We have the following lemma.
\end{longversion}
\begin{lemma}
  The shortest path to \(x\) in \(T\) changes restricted homotopy class with respect to \(x \arcto
  y\)~\(O(g)\) times.
\end{lemma}
\begin{longversion}
\begin{proof}
  Let \(A = \seq{\shortestpath_1, \shortestpath_2, \dots, \shortestpath_k}\) be the sequence of
  shortest paths to \(x\).
  Suppose \(\shortestpath_i\) and \(\shortestpath_j\) are restricted homotopic and \(j > i\). 
  Let \(r[u_i, u_j]\) denote the subpath around \(r\) from \(u_i\) to \(u_j\).
  We see \(r[u_i, u_j] \circ \shortestpath_j \circ \rev(\shortestpath_i)\) forms a disk in
  \(\Sigma_o\).
  Because the shortest path from each vertex \(u_i\) to \(x\) is unique, the shortest paths in this
  set are pairwise non-crossing.
  In particular, no shortest path between \(\shortestpath_i\) and \(\shortestpath_j\) can lie in
  another restricted homotopy class, and each restricted homotopy class is represented by a contiguous
  subsequence of \(A\).
  Let \(A'\) be the subsequence containing the first example of each homotopy class.
  Extend the beginning of each path in \(A'\) backward along \(r\) to \(u_1\);
  these extended paths still do not cross.
  The maximum number of pairwise non-crossing, non-homotopic paths in~\(\Sigma_o\) with common
  endpoints is \(O(g)\) (see Chambers~\etal~\cite[Lemma 2.1]{ccelw-scsih-08}), and this same bound
  applies to the size of \(A'\).
\end{proof}
\end{longversion}

The general strategy for bounding the running time of our algorithm is to charge various
interactions with darts to changes in the restricted homotopy class of their endpoints.
We begin by discussing pivots from stages 1 and 3.
The argument for the first case of the special pivot is similar.
Suppose we decide to pivot dart \(d^+ = x^+ \arcto y = o^+ \fencewith p^+\) into the
holiest tree \(T\) while removing dart \(d^- = x^- \arcto y = o^- \fencewith p^-\).
Recall, a vertex~\(x\) is purple if it turns from red to blue during a pivot.
We may interact with every dart \(x \arcto y\) where \(x\) is purple and \(y\) is not and their
reversals.
Cycle \(\fundcycle(T, d^+)\) is non-contractible in the surface \(\Sigma - r\).
Therefore, for each of these darts \(x \arcto y\), the restricted homotopy class of \(x\) relative
to \(x \arcto y\) changes.
We conclude that we spend~\(O(gn)\) time total interacting with individual darts during pivots in
stages 1 and~3.

We now bound the time spent checking darts for pivots during stages 1 and 3.
Suppose we check a dart~\(d = x \arcto y\).
After the check, dart~\(d\) has been passed.
It will either continue to be passed until the algorithm terminates, it is no longer active, or its
unperturbed slack decreases at the end of a round.
However, we can only decrease dart \(x \arcto y\)'s unperturbed slack at most \(\cost(x \arcto y) +
\cost(y \arcto x)\) times before we have to deactivate it and start decreasing the slack of its
reversal~\(y \arcto x\).
Deactivating~\(d\) requires \(x\) being purple, but not~\(y\), during a pivot, which requires a
change in the restricted homotopy class of \(x\) relative to \(x \arcto y\) as described above.
Overall, we spend at most~\(O(gL)\) time doing pivot checks during stages 1 and 3.

All other interactions with individual darts occur when preparing for or executing stage 2 of a
round, and in every case the \(visited\) boolean is switched for one of the darts' endpoints.
We will bound the number of times the \(visited\) booleans are set to \(\textsc{True}\) at any point
in the algorithm.
Suppose we set \(visited(p)\) to \(\textsc{True}\) and let \(succ(p) = x \arcto y = o \fencewith
p\).
If the setting occurs during stage 1 or stage 3, then the restricted homotopy class of \(x\)
relative to \(x \arcto y\) changes as discussed above.
Now suppose the setting occurs during the special pivot or during a pivot in stage 2 where dart
\(d^+ = x^+ \arcto y = o^+ \fencewith p^+\) pivots into holiest tree \(T\) while dart \(d^- = x^-
\arcto y = o^- \fencewith p^-\) leaves the tree.
Recall, \(p\) lies on a walk in the cut graph from~\(o^+\) to~\(p^-\).
Every face along this walk, including~\(p\), is enclosed by~\(\fundcycle(T, d^+)\).
The restricted homotopy class of \(x\) relative to \(x \arcto y\) changes, because the difference
between the old shortest path to \(x\) and the new shortest path to \(x\) encloses \(o\) but not
\(r\).
Finally, suppose \(visited(p)\) is set to \(\textsc{True}\) during stage 2 while checking for pivots
along~\(\cutdart_q\).
Similar to above,~\(visited(p)\) cannot be set again until~\(succ(p)\) is deactivated during stage
2,~\(succ(p)\) is deactivated during stage 1 or 3, or the slack of~\(succ(p)\) is decremented at the
end of the round so \(visited(p)\) can be unset at the end of stage 1.
The first case requires a change of restricted homotopy class as described for when \(visited(p)\)
is set during a pivot, the second case requires a change of restricted homotopy class as described
for darts touched during stages 1 and 3, and the third case can only occur \(\cost(x \arcto y) +
\cost(y \arcto x)\) times before a deactivation has to occur anyway.
The \(visited\) boolean settings, and therefore all individual dart interactions, occur~\(O(gL)\)
times total.

Finally, we must account for the time spent interacting with the reduced cut graph and the ordered
dictionary~\(\darttree\).
There is one~\(O(g)\)-time walk around the reduced cut graph per pivot and~\(O(gn)\) pivots total,
so we spend~\(O(g^2 n)\) time doing walks around the reduced cut graph.
There are~\(O(n)\) special pivots, and we can build a fresh copy of~\(\darttree\) in~\(O(g^2 \log
g)\) time after each of them for~\(O(g^2 n \log g)\) time total building these fresh copies.
A total of~\(O(gn)\) cut darts enter and leave the reduced cut graph~\cite[Section
4.2]{cce-msspe-13}, and each deletion or insertion from~\(\darttree\) takes~\(O(g \log g)\) time, so
the individual insertions and deletions during pivots take~\(O(g^2 n \log g)\) time total as well.

The overall running time of our algorithm is \(O(g(g n \log g + L))\).
We have the following theorem.

\begin{theorem}
  \label{thm:linear-time}
  Let \(G = (V, E, F)\) be a graph of size \(n\) and genus \(g\), let \(\cost :
  \dartsof{E} \to \N^+\) be a non-negative, integral dart cost function with dart costs summing to
  \(L\), and let \(r \in F\) be any face of \(G\).
  We can compute every dart entering and leaving the shortest path trees from each vertex incident
  to \(r\) in order in \(O(g(g n \log g + L))\) time.
\end{theorem}
%\begin{shortversion}
%We emphasize that we do not have time to update shortest path distances to every vertex of \(G\).
%However, generalizing an idea of Eisenstat and Klein~\cite[Section 4.9]{ek-lafms-13},
%we can maintain a subset of these distances that is sufficient for the linear-time
%algorithms in unweighted undirected graphs mentioned in the introduction.
%\end{shortversion}

\section{Applications of Linear-time Algorithm}
\label{sec:applications}

We now turn to applications of the linear-time multiple-source shortest paths algorithm described in
the previous section.

\subsection{Shortest Path Distances}
\label{subsec:applications_distances}
The above algorithm successfully computes the pivots for multiple-source shortest paths around \(r\)
in the order that they occur.
However, most applications of multiple-source shortest paths are actually concerned with at least a
subset of the shortest path distances.
Fortunately, this subset of distances is usually structured in a convenient way.

Let \(A = \seq{u_1, u_2, \dots, u_{k_1}}\) be the sequence of vertices around \(r\), and let \(B =
\seq{v_1, v_2, \dots, v_{k_2}}\) be the sequence of vertices in an arbitrary walk through \(G\).
A \EMPH{monotone correspondence} \(\correspondence\) between \(A\) and \(B\) is a set of pairs
\((u_i, v_j)\) where for each \((u_i, v_j), (u_{i'}, v_{j'}) \in \correspondence\) with \(i' \geq
i\), we have \(j' \geq j\).
Given a monotone correspondence~\(\correspondence\), we can easily modify our linear-time algorithm
to compute the unperturbed distance from \(u_i\) to \(v_j\) for every pair \((u_i, v_j)\) appearing
in \(\correspondence\) with only an \(O(k_2)\) \emph{additive} increase in the running time.
This observation is a generalization of one by Eisenstat and Klein~\cite[Theorem 4.3]{ek-lafms-13}
for planar graphs.

We store a variable \(\dist\) that is initially the unperturbed distance from \(u_1\) to \(v_1\).
The initial value for \(\dist\) can be computed in \(O(n)\) time after computing the initial holiest
tree \(T\).
As the algorithm runs, we will update \(\dist\) with the shortest path distance between some \(u_i\)
and \(v_j\).
Suppose an iteration of the algorithm has just ended and we are storing the \(u_i\) to \(v_j\)
distance.
We can compute the unperturbed distance from \(u_i\) to \(v_{j+1}\) as as \(\dist + \cost(v_j \arcto
v_{j+1}) - \uslack(v_j \arcto v_{j+1})\) and reassign \(\dist\) to that value.
We repeat this step until we have computed distances for every pair containing \(u_i\).

Now, suppose we have just performed the special pivot to move the source of \(T\) from \(u_i\) to
\(u_{i+1}\).
After the special pivot, the distance to every vertex in \(G\) from the source of \(T\) has
decreased by the distance from \(u_i\) to \(u_{i+1}\).
We decrease \(\dist\) by that amount.
Now, the unperturbed distance from \(u_{i+1}\) to some vertex \(v_j\) increases by \(1\) at the end
of each fully completed round where \(v_j\) is red.
To easily track if \(v_j\) is red, we maintain marks on edges appearing an odd number of times along
an arbitrary walk from \(u_{i+1}\) to \(v_j\).
If there are an odd number of marked edges containing active darts when we increase \(\uparameter\),
then the unperturbed distance to \(v_j\) increases by \(1\) and we increment \(\dist\).
Otherwise, \(\dist\) remains unchanged.
To maintain these marks for any pair \((u_i, v_j)\) we compute an arbitrary \((u_1, v_1)\) walk at
the beginning of the algorithm.
Every time we consider distances to the next vertex along \(B\), we flip the mark on the next edge
used in \(B\)'s walk.
Every time we perform a special pivot, we flip the mark on the edge for \(u_i \arcto u_{i+1}\).

\begin{theorem}
  \label{thm:linear-time_distances}
  Let \(G = (V, E, F)\) be a graph of size \(n\) and genus \(g\), let \(\cost :
  \dartsof{E} \to \N^+\) be a non-negative, integral dart cost function with dart costs summing to
  \(L\).
  Let \(r \in F\) be any face of \(G\) incident to vertices \(A = \seq{u_1, u_2, \dots, u_{k_1}}\)
  in order, and let \(B = \seq{v_1, v_2, \dots, v_{k_2}}\) be the sequence of vertices along an
  arbitrary walk in \(G\).
  Let \(\correspondence\) be an arbitrary monotone correspondence between \(A\) and \(B\).
  We can compute the distance from \(u_i\) to \(v_j\) for every pair \((u_i, v_j) \in
  \correspondence\) in \(O(g(g n \log g + L) + k_2)\) time.
\end{theorem}

\subsection{The Applications}

We can use the algorithm of Theorem~\ref{thm:linear-time_distances} to easily derive deterministic
linear-time algorithms for a variety of problems on unweighted undirected embedded graphs of
constant genus.
We discuss a few of these problems in this section.
For each problem, there is a published algorithm that runs in near-linear time that can use the
multiple-source shortest paths algorithm of Cabello~\etal~\cite{cce-msspe-13} in a black box fashion.
While Cabello \etal~require uniqueness of shortest paths, it is otherwise unnecessary in these
algorithms.
In every case, the set of shortest path distances required by the algorithm have the form required
by Theorem~\ref{thm:linear-time_distances} with \(k_2 = O(n)\).

We note our multiple-source shortest paths algorithm is not necessary for \(g^{O(g)} n\) or (in some
cases) \(2^{O(g)} n\) time algorithms for these problems, assuming that the input graph is unweighted
and undirected.
However, we are unaware of any publications stating this observation explicitly.
In every case, the use of our linear-time multiple-source shortest paths algorithm results in a
substantial decrease in the dependency on \(g\).
\begin{shortversion}
We now give our improvements.
\end{shortversion}

\begin{longversion}
We give one example of such an algorithm here;
the slower algorithms for the other problems use similar ideas.
\end{longversion}
Given an undirected possibly edge weighted graph \(G\) and two vertices \(s\) and \(t\), an
\(s,t\)-cut is a bipartition \((S, T)\) of the vertices such that \(s \in S\) and \(t \in T\).
The \EMPH{capacity} of an \(s,t\)-cut \((S, T)\) is the total number/total weight of all edges with
exactly one endpoint in \(S\).
The \EMPH{minimum \(s,t\)-cut} is the \(s,t\)-cut of minimum capacity.
\begin{shortversion}
Plugging our multiple-source shortest paths procedure into an algorithm of Erickson and
Nayyeri~\cite{en-mcsnc-11}, we derive the following result.
\end{shortversion}
\begin{longversion}
Chambers~\etal~\cite{cen-mcshc-09} reduce the problem of computing a minimum \(s,t\)-cut in an
undirected genus \(g\) graph to \(g^{O(g)}\) computations of minimum \(s',t'\)-cuts in planar graphs
of size \(O(gn)\).
If the input graph is unweighted, then these computations can each be done in \(O(gn)\) time each
using an algorithm of Weihe~\cite{w-edstp-97}, yielding a \(g^{O(g)} n\) time algorithm for minimum
\(s,t\)-cut in unweighted undirected graphs.
We now turn more efficient algorithms that use our linear-time multiple-source shortest paths
algorithm instead.

Instead of using the algorithm given above, we can compute minimum \(s,t\)-cuts in a genus \(g\)
surface embedded graph using an algorithm of Erickson and Nayyeri~\cite{en-mcsnc-11} that runs
\(2^{O(g)}\) instances of multiple-source shortest paths in a derived graph of size \(2^{O(g)} n\)
and genus \(2^{O(g)}\).
\end{longversion}
\begin{theorem}
  \label{thm:linear_st_cut}
  Let \(G = (V, E, F)\) be an unweighted undirected graph of size \(n\) and genus~\(g\),
  and let \(s,t \in V\).
  There exists a deterministic algorithm that computes a minimum \(s,t\)-cut of \(G\) in \(2^{O(g)}
  n\) time.
\end{theorem}

A \EMPH{global minimum cut} is an \(s,t\)-cut of minimum capacity across all choices of distinct
vertices \(s\) and \(t\).
\begin{shortversion}
We combine our algorithm with algorithms by Erickson~\etal~\cite{efn-gmcse-12}, Erickson and
Nayyeri~\cite{en-mcsnc-11}, and Chang and Lu~\cite{cl-cgpgl-13}.
\end{shortversion}
\begin{longversion}
We can compute global minimum cuts in linear time as well using an algorithm of Erickson
\etal~\cite{efn-gmcse-12} with some slight modifications.
We replace their use of an algorithm by Chambers~\etal~\cite{cen-mcshc-09} for computing minimum
weight homologous subgraphs (with coefficients in \(\Z_2\)) with the algorithm of Erickson and
Nayyeri~\cite{en-mcsnc-11} for the same problem, plugging in our multiple-source shortest paths
algorithm when needed by Erickson and Nayyeri.
Again, this results in \(2^{O(g)}\) uses of our multiple-source shortest paths algorithm in graphs of
size \(2^{O(g)} n\) of genus \(2^{O(g)}\).
Also, we replace Erickson~\etal's~\cite{efn-gmcse-12} constant number of uses of an algorithm by
Łącki and Sankowski~\cite{ls-mcsc-11} for global minimum cut in planar graphs with a linear time
algorithm for the same problem in unweighted planar graphs by Chang and Lu~\cite{cl-cgpgl-13}.
\end{longversion}
\begin{theorem}
  \label{thm:linear_global_cut}
  Let \(G = (V, E, F)\) be an unweighted undirected graph of size \(n\) and genus~\(g\).
  There exists a deterministic algorithm that computes a global minimum cut of \(G\) in \(2^{O(g)}
  n\) time.
\end{theorem}

Let graph \(G\) be embedded in surface \(\Sigma\).
A \EMPH{non-separating cycle}~\(\cycle\) in \(G\) is one for which \(\Sigma - \cycle\) is connected.
A shortest non-separating or non-contractible cycle is a non-separating or non-contractible cycle
with a minimum number of edges or total cost if the edges have costs.
\begin{longversion}
We can compute shortest non-separating and non-contractible cycles by simply substituting our
multiple-source shortest paths algorithm into the algorithms of Erickson~\cite{e-sncds-11} and
Fox~\cite{f-sntcd-13}.
Erickson uses \(O(g)\) instances of multiple-source shortest paths in derived graphs of size
\(O(n)\) and genus \(O(g)\) to compute a shortest non-separating cycle.
Fox uses \(O(g^2 + b)\) instances in derived graphs of size \(O(n)\) and genus \(O(g)\) to compute a
shortest non-contractible cycle when the embedding has \(b\) boundary components.
Note that these algorithms are actually designed for directed graphs with dart costs, but the
undirected case follows from the obvious reduction.
For these two problems, there is no other linear time algorithm known for directed graphs with small
integer dart costs.
\end{longversion}
\begin{shortversion}
For these two problems, there was no linear time algorithm known for directed graphs with small
integer dart costs.
However, there were~\(2^{O(g)} n\) time algorithms for the undirected case as mentioned above.
We combine our algorithm with algorithms by Erickson~\cite{e-sncds-11} and Fox~\cite{f-sntcd-13}.
\end{shortversion}
\begin{longversion}
However, a running time of \(2^{O(g)} n\) is achievable in unweighted undirected graphs by applying
Weihe's~\cite{w-edstp-97} minimum \(s,t\)-cut algorithm as a subroutine in Fox's~\cite{f-sntcd-13}
modification of Kutz's~\cite{k-csnco-06} original near-linear time algorithm for these problems.
\end{longversion}
\begin{theorem}
  \label{thm:linear_non-trivial}
  Let \(G = (V, E, F)\) be a graph of size \(n\) and genus \(g\), embedded in a surface with \(b\)
  boundary components, and let \(\cost : \dartsof{E} \to \N^+\) be a positive, integral dart cost
  function with dart costs summing to \(L\).
  We can compute a shortest non-separating cycle in \(G\) in \(O(g^2(g n \log g + L))\) time and a
  shortest non-contractible cycle in \(G\) in \(O((g^2+b)g(g n \log g + L))\) time.
\end{theorem}
%\note{Kyle: We could probably do better for non-contractible, but then we'd need linear time
%versions of a few other things including shortest cycle homotopic to a boundary component.}

Finally, a \EMPH{homology basis} is a maximal collection of cycles belonging to linearly independent
homology classes with coefficients in \(\Z_2\).
A \EMPH{shortest homology basis} is one in which the total number of edges or edge costs is
minimized.
We combine our algorithm with one by Borradaile~\etal~\cite{bcfn-mchbs-17}.
\begin{longversion}
We can compute a shortest homology basis in an unweighted undirected graph by plugging our
multiple-source shortest paths algorithm into the algorithm of Borradaile \etal~\cite{bcfn-mchbs-17}.
Their algorithm uses~\(O((g + b)^2)\) instances of multiple-source shortest paths in a derived graph
of size \(O(n)\) and genus \(O(g+b)\) when the embedding has \(b\) boundary faces.
\end{longversion}
\begin{theorem}
  \label{thm:linear_homology-basis}
  Let \(G = (V, E, F)\) be an unweighted undirected graph of size \(n\) and genus~\(g\), embedded in
  a surface with \(b\) boundary components.
  There exists a deterministic algorithm that computes a minimum homology basis of \(G\) in
  \(O((g+b)^4 n \log (g + b))\) time. 
\end{theorem}

\ifx\compilelong\undefined
\begin{acks}
\else
\paragraph*{Acknowledgements.}
\fi
The authors would like to thank Sergio Cabello, Erin W. Chambers, and Shay Mozes for many helpful
discussions.
\ifx\compilelong\undefined
They would also like to thank the anonymous reviewers for their helpful comments on the writing.
\end{acks}
\else
They would also like to thank the anonymous reviewers of the conference version of this paper for
their helpful comments on the writing.
\fi

% Begin bbl nonsense
\begin{longversion}
\bibliographystyle{abbrv}
\end{longversion}
\begin{shortversion}
\bibliographystyle{ACM-Reference-Format}
\balance
\end{shortversion}
\bibliography{jeff-bib/topology,jeff-bib/optimization,jeff-bib/data-structures,new_refs}

\end{document}